\documentclass[aps,english,floatfix,superscriptaddress,tightenlines,twocolumn,nofootinbib]{revtex4-2}
\usepackage[margin = 0.7in]{geometry}

\usepackage[indentafter]{titlesec}
\titleformat{\section}{\normalfont\small\bfseries\centering}{\thesection.}{0.5em}{\MakeUppercase}
\titlespacing*{\section}{0pt}{20 pt}{10 pt}

\renewcommand{\thesection}{\arabic{section}}

\usepackage{amssymb}
\usepackage{amsmath,bm}
\usepackage{amssymb}
\usepackage{graphicx}
\usepackage{amsfonts}         
\usepackage{fancybox}

\usepackage[usenames,dvipsnames]{xcolor}

\definecolor{Mathematica1}{rgb}{0.368417, 0.506779, 0.709798}
\definecolor{Mathematica2}{rgb}{0.880722, 0.611041, 0.142051}
\definecolor{Mathematica3}{rgb}{0.560181, 0.691569, 0.194885}

\usepackage[pagebackref,  
bookmarks={false}, pdfauthor={Entanglement Bootstrap Collaboration}, pdftitle={Yay, physics!}]{hyperref}
\hypersetup{colorlinks=true, 
linkcolor=Mathematica1,
citecolor=Mathematica3,
filecolor=OliveGreen, 
urlcolor=Mathematica2,
filebordercolor={.8 .8 1}, 
urlbordercolor={.8 .8 0}}
\usepackage{soul}
\setstcolor{Red}

\definecolor{darkgreen}{rgb}{0,0.4,0}
\definecolor{darkred}{rgb}{0.4,0,0}
\definecolor{darkblue}{rgb}{0,0,0.4}
\definecolor{lightblue}{rgb}{.6,.6,0.9}

\definecolor{uglybrown}{rgb}{0.8,  0.7,  0.5}

\definecolor{palatinatepurple}{rgb}{0.41, 0.16, 0.38}
\definecolor{celebrationcolor}{rgb}{0.75,  0.0,  0.9}

\definecolor{shadecolor}{rgb}{0.90,0.90,0.90}
\definecolor{DVcolor}{rgb}{0.95,  0.5,  0.2}
\definecolor{lightbluemuons}{rgb}{0.0,.65,1.0}

\definecolor{chartreuse}{rgb}{0.70, 1.00, 0.00}

\usepackage{tikz}
\usetikzlibrary{shapes}
\usetikzlibrary{trees}
\usetikzlibrary{matrix,arrows} 				
\usetikzlibrary{positioning}				
\usetikzlibrary{calc,through}				
\usetikzlibrary{decorations.pathreplacing}  
\usepackage[tikz]{bclogo} 					
\usepackage{pgffor}							

\usetikzlibrary{decorations.markings}
\usetikzlibrary{intersections}				

\usetikzlibrary{arrows,decorations.pathmorphing,backgrounds,positioning,fit,petri,automata,shadows,calendar,mindmap, graphs}
\usetikzlibrary{arrows.meta,bending}

\tikzset{
    vector/.style={decorate, decoration={snake}, draw},
    fermion/.style={postaction={decorate},
        decoration={markings,mark=at position .55 with {\arrow{>}}}},
    fermionbar/.style={draw, postaction={decorate},
        decoration={markings,mark=at position .55 with {\arrow{<}}}},
    fermionnoarrow/.style={},
    gluon/.style={decorate,
        decoration={coil,amplitude=4pt, segment length=5pt}},
    scalar/.style={dashed, postaction={decorate},
        decoration={markings,mark=at position .55 with {\arrow{>}}}},
    scalarbar/.style={dashed, postaction={decorate},
        decoration={markings,mark=at position .55 with {\arrow{<}}}},
    scalarnoarrow/.style={dashed,draw},
	vectorscalar/.style={loosely dotted,draw=black, postaction={decorate}},
}

\def\centerarc[#1](#2)(#3:#4:#5)
    { \draw[#1] ($(#2)+({#5*cos(#3)},{#5*sin(#3)})$) arc (#3:#4:#5); }

\usepackage{enumitem}

\usepackage{slashed}

\usepackage[font=small,labelfont=bf]{caption}
\DeclareCaptionFont{tiny}{\tiny}
\usepackage{epstopdf}
\usepackage{wasysym}

\usepackage[yyyymmdd,hhmmss]{datetime}   
\usepackage{framed}  
 \usepackage{mdframed}
 \usepackage{wrapfig}
 \usepackage{yfonts}  

\newmdenv[%
        backgroundcolor=lightgray,
    linecolor=black,
    outerlinewidth=2pt,
]{boxedandshaded}

\numberwithin{equation}{section}

\renewcommand{\theequation}{\arabic{section}.\arabic{equation}}

\newlength{\extraspace}
\newlength{\extraspaces}
\def\be{\begin{equation}}
\def\ee{\end{equation}}

\newcommand{\bea}{\begin{eqnarray}}
\newcommand{\eea}{\end{eqnarray}}

\def\Tr{{{\rm Tr}}}

\def\Im{{\rm Im\hskip0.1em}}

\def\bra#1{\left\langle#1\right|}
\def\ket#1{\left|#1\right\rangle}

\def\CF{{\cal F}}

\def\II{\relax{I\kern-.10em I}}

\def\IB{\relax{\rm I\kern-.18em B}}

\def\ID{\relax{\rm I\kern-.18em D}}
\def\IE{\relax{\rm I\kern-.18em E}}
\def\IF{\relax{\rm I\kern-.18em F}}
\def\IG{\relax\hbox{$\inbar\kern-.3em{\rm G}$}}
\def\IGa{\relax\hbox{${\rm I}\kern-.18em\Gamma$}}
\def\IH{\relax{\rm I\kern-.18em H}}
\def\II{\relax{\rm I\kern-.18em I}}
\def\IK{\relax{\rm I\kern-.18em K}}

\def\inbar{\,\vrule height1.5ex width.4pt depth0pt}

\def\simgt{\hskip0.05in\relax{ 
\raise3.0pt\hbox{ $>$
{\lower5.0pt\hbox{\kern-1.05em $\sim$}} }} \hskip0.05in}

\def\lp10{\ell_p^{10}}
\def\lp11{\ell_p^{11}}
\def\R11{R_{11}}

\def\frac#1#2{{#1 \over #2}}

\newdimen\tableauside\tableauside=1.0ex
\newdimen\tableaurule\tableaurule=0.4pt
\newdimen\tableaustep
\def\phantomhrule#1{\hbox{\vbox to0pt{\hrule height\tableaurule width#1\vss}}}
\def\phantomvrule#1{\vbox{\hbox to0pt{\vrule width\tableaurule height#1\hss}}}
\def\sqr{\vbox{%
  \phantomhrule\tableaustep
  \hbox{\phantomvrule\tableaustep\kern\tableaustep\phantomvrule\tableaustep}%
  \hbox{\vbox{\phantomhrule\tableauside}\kern-\tableaurule}}}
\def\squares#1{\hbox{\count0=#1\noindent\loop\sqr
  \advance\count0 by-1 \ifnum\count0>0\repeat}}
\def\tableau#1{\vcenter{\offinterlineskip
  \tableaustep=\tableauside\advance\tableaustep by-\tableaurule
  \kern\normallineskip\hbox
    {\kern\normallineskip\vbox
      {\gettableau#1 0 }%
     \kern\normallineskip\kern\tableaurule}%
  \kern\normallineskip\kern\tableaurule}}
\def\gettableau#1 {\ifnum#1=0\let\next=\null\else
  \squares{#1}\let\next=\gettableau\fi\next}

\def\({\left(}
\def\){\right)}

\def\ii{{\bf i}}

\def\lsim{\mathrel{\mathstrut\smash{\ooalign{\raise2.5pt\hbox{$<$}\cr\lower2.5pt\hbox{$\sim$}}}}}
\def\gsim{\mathrel{\mathstrut\smash{\ooalign{\raise2.5pt\hbox{$>$}\cr\lower2.5pt\hbox{$\sim$}}}}}

\def\overleftrightarrow#1{\vbox{\ialign{##\crcr
     $\leftrightarrow$\crcr\noalign{\kern-0pt\nointerlineskip}
     $\hfil\displaystyle{#1}\hfil$\crcr}}}
     
     \def\overleftarrow#1{\vbox{\ialign{##\crcr
     $\leftarrow$\crcr\noalign{\kern-0pt\nointerlineskip}
     $\hfil\displaystyle{#1}\hfil$\crcr}}}

\def\gU{\textsf{U}}

\newif{\ifeq}           
\eqtrue                 
\newcounter{lecturecounter}

\usepackage{comment}
\usepackage{amsmath} 
\usepackage{mathtools}
\usepackage{amssymb}

\usepackage{amsthm}
\usepackage{thmtools}

\theoremstyle{plain}
\newtheorem{theorem}{Theorem}[section] 

\newtheorem{definition}[theorem]{Definition} 

\theoremstyle{plain}

\theoremstyle{plain}
\newtheorem{lemma}[theorem]{Lemma}
\theoremstyle{plain}

\theoremstyle{plain}

\theoremstyle{plain}
\newtheorem{Fact}[theorem]{Fact}

\definecolor{XLgreen}{RGB}{34,139,34}
\definecolor{JMblue}{RGB}{25,25,125}
\definecolor{BSorange}{RGB}{140,50,0}

\usepackage[export]{adjustbox}
\usepackage{graphicx}

\def\ins{\mathcal{I}}

\def\calF{\mathcal{F}}
\def\calH{\mathcal{H}}

\allowdisplaybreaks
\allowbreak

\begin{document}
\title{\bfseries\Large Strict area law entanglement versus chirality}
\author{Xiang Li}
\affiliation{Department of Physics, University of California at San Diego, La Jolla, CA 92093, USA}
\author{Ting-Chun Lin}
\affiliation{Department of Physics, University of California at San Diego, La Jolla, CA 92093, USA}
\affiliation{Hon Hai Research Institute, Taipei, Taiwan}
\author{John McGreevy}
\affiliation{Department of Physics, University of California at San Diego, La Jolla, CA 92093, USA}
\author{Bowen Shi}
\affiliation{Department of Physics, University of California at San Diego, La Jolla, CA 92093, USA}
	\affiliation{Department of Computer Science, University of California, Davis, CA 95616, USA}

\begin{abstract}
Chirality is a property of a gapped phase of matter in two spatial dimensions that can be manifested through non-zero thermal or electrical Hall conductance.   
In this paper, we prove two no-go theorems that forbid 
such chirality for
a quantum state in a finite dimensional local Hilbert space with strict area law entanglement entropies. We also show that the finite dimensional local Hilbert space condition can be relaxed to the condition that state has finite local entanglement entropies.
As a crucial ingredient in the proofs, we introduce a new 
quantum information-theoretic primitive 
called \emph{instantaneous modular flow}, which has many other potential applications.

\end{abstract}
\maketitle

\section{Introduction}

A quantum phase of matter can be regarded as a family of quantum many-body systems that share the same universal properties. In this paper, we focus on 2+1D chiral gapped phases of matter, in which a system can be described by a local many-body Hamiltonian with a bulk energy gap in two spatial dimensions and have non-zero thermal or electrical Hall conductance. 

Both thermal and electrical Hall conductance characterize chiral degrees of freedom of the system and are universal properties of the phase. When a gapped quantum system is put on a two-dimensional disk, at a low temperature $T$, it is possible that there is a non-vanishing net energy current $I^E(T)$ flowing in a certain direction along the edge, making the system chiral. This is known as the thermal Hall effect. The thermal Hall conductance is of the form $dI^E(T)/dT = \pi T c_{-}/6$, where $c_{-}$ is called the chiral central charge \cite{Kane_1997,Cappelli:2001mp,Kitaev:2005hzj}.
This thermal Hall effect is robust against any local perturbations to the system. The chiral central charge, which describes the difference between the number of edge degrees of freedom propagating clockwise and counterclockwise, is universal and takes the same value for all the systems within the phase. Electrical Hall conductance describes the linear response of particles with $\gU(1)$ electric charges on a two-dimensional plane in the presence of an electric and magnetic field. 
The conductance of the induced $\gU(1)$ current $J_x$ with the electric field $E_y$ is of the form
$\sigma_{xy} = J_x/E_y = \nu/(2\pi)$, 
with $\nu$ being a rational number \cite{Laughlin_1981,Halperin_1982}. $\sigma_{xy}$ is also universal in the phase and robust against local perturbations that do not break the $\gU(1)$ symmetry. 

It is natural to ask: which systems can or cannot have a non-zero thermal or electrical quantum Hall effect? Ref.~\cite{Kapustin_2019,Zhang_2022,Kapustin_2020} answer the question in terms of Hamiltonians of the systems: A system with non-zero thermal or electrical Hall conductance cannot be described by a commuting projector Hamiltonian in a finite-dimensional Hilbert space.  
A common method of studying a gapped phase of matter is to construct a commuting projector Hamiltonian of a system in the phase, so that the system is exactly solvable. The theorems in \cite{Kapustin_2019,Zhang_2022,Kapustin_2020} show that such a method is impossible for chiral gapped phases, unless one forgoes some requirements such as finite-dimensional Hilbert space, as in \cite{DeMarco:2021erp}.  

The result presented in this paper concerns a different approach to studying phases of matter. Instead of starting with a Hamiltonian that describes a system within the phase, one can also directly work with a suitable quantum state that encodes the universal properties of the phase. Such a state is called a \emph{representative state} of the phase. There has been a growing number of works in this approach \cite{kitaev2006topological,levin2006detecting,Li-Haldane2008,shi2020fusion,knots-paper,Kim2021,Kim:2021tse,fan2022extracting,
Cian2022,Kobayashi2023,Lin2023,q-cross-ratio2024,emergence-of-virasoro2024}. In particular, chiral central charge and electric Hall conductance of a chiral gapped phase can be extracted just from a single representative state of the phase \cite{Kim:2021tse,Kim2021,fan2022extracting}. Among these works of studying phases of matter from a quantum state, there is a common feature that local entanglement properties of the state play an important role. Therefore, to study chiral gapped phases of matter, it is of great importance to explore the interplay between chirality and local entanglement properties in representative states of the phases. Motivated by this, we ask the following question: \emph{What are the general local entanglement properties a quantum state should have, so that it can have a non-zero $c_{-}$ or $\sigma_{xy}$ and hence be a representative state of a chiral gapped phase?} In this paper, we provide an answer to this question, which is summarized into two no-go theorems: 
\emph{A representative state of a chiral gapped phase cannot have strict area law entanglement entropies if the state is constructed in a tensor product Hilbert space with finite local dimensions.} 

\section{Preliminaries} 
\label{sec:setup}
The precise statements of the no-go theorems involve three aspects of a quantum many-body state $\ket\Psi$, namely (1) its setup in a tensor product Hilbert space with finite local dimension, (2) strict area law condition (more precisely, the bulk {\bf A1} condition), (3) computation of $c_{-}$ and $\sigma_{xy}$ from a representative state of a chiral gapped phase. We now introduce them in detail. 

(1) \emph{Setup}: We shall consider a two-dimensional quantum many-body state $\ket\Psi$ in a \emph{tensor product Hilbert space with finite local dimension}. We explicitly unpack these words as follows: The state $\ket\Psi$ describes a quantum many-body system on a two-dimensional lattice. The local degrees of freedom live on the lattice sites $v$ with finite dimensional Hilbert spaces $\calH_v$. The global Hilbert space for $\ket{\Psi}$ is a tensor product $\calH = \otimes_v \calH_v$\footnote{If the local degrees of freedom are fermionic, then the tensor product is $\mathbb{Z}_2$ graded.}. Throughout this paper, we shall exclusively use $\ket\Psi$ to denote a quantum state under such a setup. 

(2) \emph{Strict area law and bulk {\bf A1}}: The local entanglement properties of $\ket\Psi$ are encoded in its reduced density matrices. Consider a region $A$, which is a collection of a finite number of sites in the lattice; the Hilbert space can be factored as $\calH = \calH_A \otimes \calH_{\overline{A}}$, where $\calH_X = \otimes_{v\in X}\calH_v$ and $\overline{A}$ denotes the complement of region $A$. The reduced density matrix $\rho_A \equiv \Tr_{\overline{A}}\ket{\Psi}\bra{\Psi}$ is defined via tracing out the degrees of freedom in $\calH_{\overline{A}}$. If $\ket{\Psi}$ represents the infrared fixed-point of the renormalization group flow in a gapped phase of matter, it is expected to satisfy \emph{strict area law}. This explicitly means the entanglement entropy $S(\rho_A)\equiv -\Tr(\rho_A \ln \rho_A)$ of the reduced density matrix $\rho_A$ on a simply connected region $A$ is of the form 
\begin{equation}\label{eq:strict-area-law}
    S(\rho_A) = \alpha |\partial A| - \gamma,
\end{equation}
where $\alpha,\gamma$ are two region-independent constants and $|\partial A|$ denotes the size of the boundary of $A$. For general representative states of the phase, there can be subleading contributions in addition to Eq.~\eqref{eq:strict-area-law}, which usually decay as $|\partial A|$ increases. For different representative states of the phase, $\alpha$ might be different, while $\gamma$, known as the topological entanglement entropy \cite{kitaev2006topological,levin2006detecting}, usually takes the same value and hence is a universal property of the phase. 

\begin{figure}[!h]
    \centering
    \includegraphics[width=0.5\columnwidth]{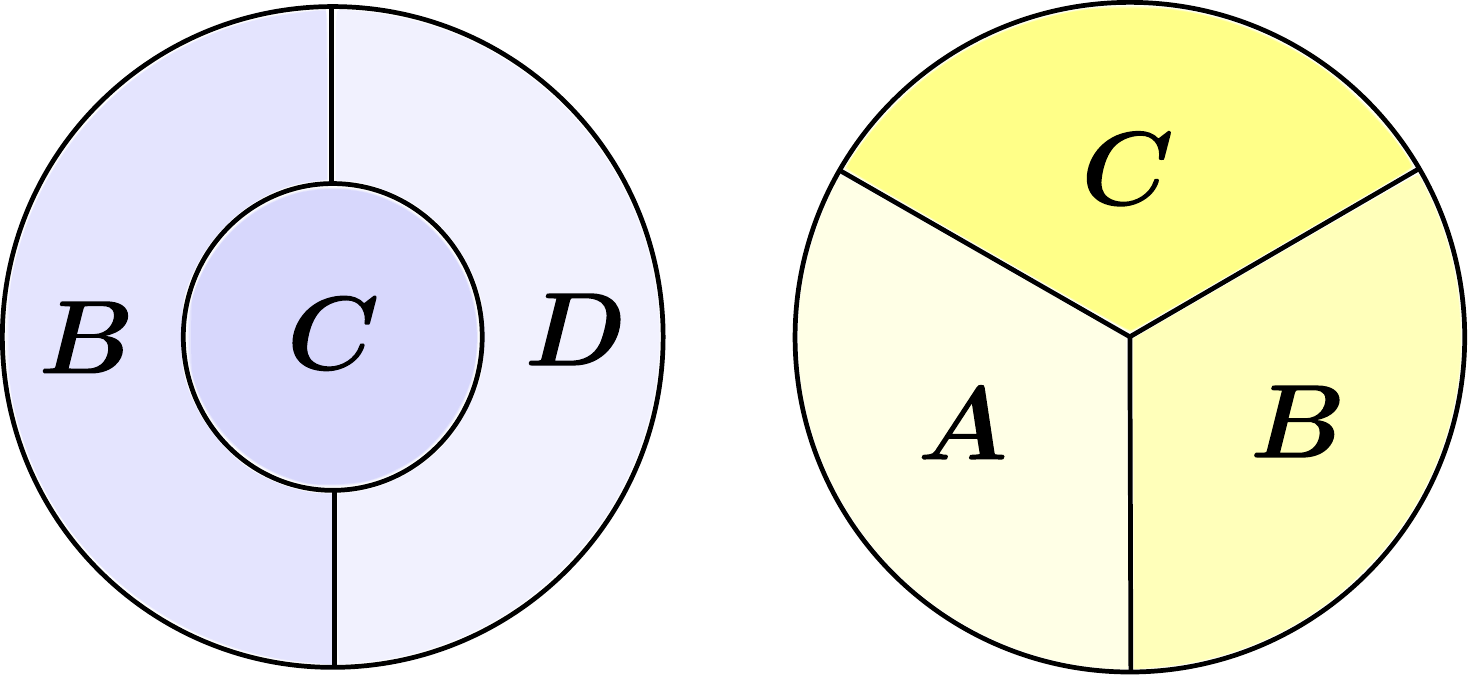}
    \caption{Tripartite regions on the lattice. (Left) Regions for bulk {\bf A1}. (Right) Regions for computing $c_{-}$ and $\sigma_{xy}$.}
    \label{fig:bulk-A1-J}
\end{figure}

In fact, the no-go theorems we shall discuss later concern quantum states that satisfy a condition called \emph{bulk {\bf A1}} introduced in \cite{shi2020fusion}, which relaxes the strict area law condition. Explicitly, we say $\ket\Psi$ satisfies the bulk {\bf A1} if for any disk-like region with partition $BCD$ topologically equivalent to Fig.~\ref{fig:bulk-A1-J} (Left), the entanglement entropies from $\ket\Psi$ on these regions satisfy  
\begin{equation}\label{eq:A1}
    (S_{BC}+S_{CD}-S_B-S_D)_{\ket\Psi} = 0,
\end{equation}
where $(S_X)_{\ket\Psi}\equiv S(\rho_X)$. With Eq.~\eqref{eq:strict-area-law} one can check that a strict area law state satisfies bulk {\bf A1}. On the other hand, {\bf A1} states include more than strict area law states, since {\bf A1} could still hold if there are additional contributions to Eq.~\eqref{eq:strict-area-law}
\footnote{In general, bulk {\bf A1} allows any additional short-range entanglement contributions in the form of sum of local terms along $\partial A$ as discussed in \cite{Grover:2011fa}. 
An explicit state satisfying {\bf A1} but not strict area law can be obtained by starting with a state with strict area law and tensoring in EPR pairs $\sqrt{p_{uv}} \ket{0_u0_v} + \sqrt{1-p_{uv}} \ket{1_u1_v}$ of varying strength $p_{uv}$ along each link $<uv>$ of the lattice. The whole state satisfies {\bf A1} still, but the area law coefficient will vary in space according to the choices of $p$.
Moreover, sharp corners of the region can make additional contributions to Eq.~\eqref{eq:strict-area-law}, which are consistent with {\bf A1}.
In fact,
we show in \cite{Corner_ent} that such a corner contribution is a necessary consequence of $c_- \neq 0$.
However, as a consequence of the no-go theorem~\ref{thm:no-go-thermal-intro}, this violation of the exact area law cannot be the only one that results from $c_- \neq 0$ --  since a corner contribution is consistent with {\bf A1}, further violation is required.}.
An important fact about a state that satisfies bulk {\bf A1} is 
\begin{Fact}\label{fact:markov}
    If $\ket{\Psi}$ satisfies bulk {\bf A1}, then for any tripartite region $ABC$ in the bulk that is topologically equivalent to the one shown in Eq.~\eqref{eq:CMI-0}, 
    \begin{equation}\label{eq:CMI-0}    
    I(A:C|B)_{\ket\Psi}=0,
     \quad\includegraphics[width=0.25\columnwidth,,valign=c]{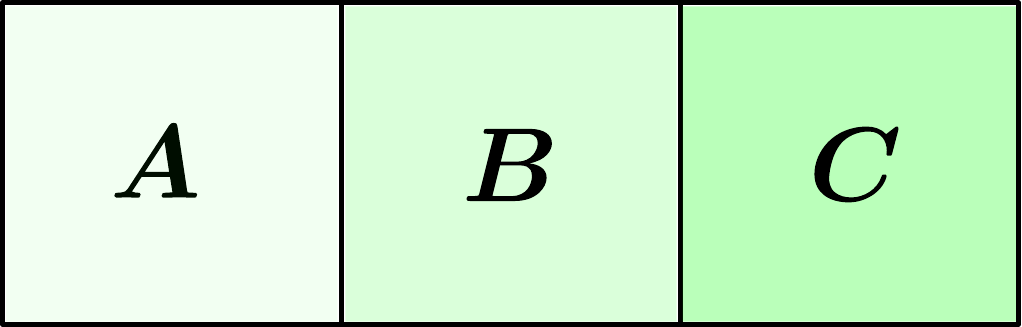}
    \end{equation}
    where $I(A:C|B)_{\ket\Psi} \equiv (S_{AB}+S_{BC}-S_B-S_{ABC})_{\ket{\Psi}}$. 
\end{Fact}
\noindent A state that satisfies Eq.~\eqref{eq:CMI-0} is called \emph{Markov} on $ABC$. We utilize this fact many times later in the paper.

(3) \emph{Computation of chiral central charge $c_{-}$ and electric Hall conductance $\sigma_{xy}$}:  
Previously, we have introduced chiral central charge as a property about the edge energy current. In fact, it is also a property of the bulk state. If $\ket\Psi$ is a suitable representative state of a chiral gapped phase with non-vanishing thermal Hall conductance, one can compute $c_{-}$ using bulk reduced density matrices of $\ket\Psi$ through\footnote{Similarly to the situation for topological entanglement entropy \cite{Williamson:2018zig,Kato:2019cdi}, 
there are pathological states where $J$ is not associated with a meaning of chiral central charge \cite{Gass:2024dok}. 
These examples also violate {\bf A1}. Nevertheless, $J\neq 0$ is an indication of chirality of the state regardless whether it is related to the thermal Hall conductance, and our result indicates such a chirality is not compatible with strict area law state. 
}
\begin{equation}\label{eq:J-c}
     J(A,B,C)_{\ket{\Psi}} \equiv \ii \bra{\Psi} [K_{AB},K_{BC}] \ket{\Psi} = \frac{\pi c_{-}}{3},
\end{equation}
where $\ii = \sqrt{-1}$ and $K_{\bullet} = - \ln \rho_{\bullet}$ is called modular Hamiltonian whose support indicated by the subscript is on region shown in Fig.~\ref{fig:bulk-A1-J} (Right). To compute electric Hall conductance, $\ket\Psi$ shall have an \emph{on-site $\gU(1)$ symmetry} $U(t) = \prod_v U_v(t)$, i.e. $U(t)\ket\Psi \propto \ket\Psi$. $U_v(t) = e^{\ii Q_v t}$ is generated by $Q_v$, which is supported on a lattice site $v$ and measures the amount of $\gU(1)$ charges on $v$.
Then, the local charge operator on a region $X$ is defined as $Q_X = \sum_{v\in X}Q_v$.
On such a $\gU(1)$ symmetric state $\ket\Psi$, it was proposed in \cite{fan2022extracting} that the Hall conductance can be computed via 
\begin{equation}\label{eq:Sigma-sigma}
    \Sigma(A,B,C)_{\ket\Psi} \equiv \frac{\ii }{2}\langle \Psi| [K_{AB},Q_{BC}^2] |\Psi\rangle = \sigma_{xy},
\end{equation} 
using the same tripartite region shown in Fig.~\ref{fig:bulk-A1-J} (Right). We list several properties of $J$ and $\Sigma$ in App.~\ref{app:modular-commutators}\footnote{We remark that to show $\Sigma(A,B,C)_{\ket\Psi}$ is a topological quantity, Ref.~\cite{fan2022extracting} assumed a so-called cluster property of the state. We discover that just the bulk {\bf A1} is enough. The proof of this fact is given in App.~\ref{app:sigma-property}.}. 

\section{Statement of the no-go theorems}
Now, we are in a position to present the no-go theorems, which reveal a tension among these three aspects we introduced at the beginning of Sec.~\ref{sec:setup}.  
\begin{theorem}\label{thm:no-go-thermal-intro}
    There is no quantum many-body state $\ket{\Psi}$ in a tensor product Hilbert space with finite local dimension satisfying bulk {\bf A1} with non-zero $c_{-}$ from $J(A,B,C)_{\ket\Psi}$ using Eq.~\eqref{eq:J-c}. 
\end{theorem}
\begin{theorem}\label{thm:no-go-electric-intro}
    There is no quantum many-body state $\ket{\Psi}$ in a tensor product Hilbert space with finite local dimension satisfying bulk {\bf A1} and having an on-site $\gU(1)$ symmetry and non-zero $\sigma_{xy}$ from $\Sigma(A,B,C)_{\ket\Psi}$ using Eq.~\eqref{eq:Sigma-sigma}. 
\end{theorem}

We remark that the condition of finite local Hilbert space dimension in both Theorems~\ref{thm:no-go-thermal-intro} and~\ref{thm:no-go-electric-intro} can be relaxed into finite local entanglement entropies. We leave the discussion of the generalized version in App.~\ref{app:general-thm}.

\section{Instantaneous modular flow}
\label{sec:IMF}
We now introduce a novel quantum information-theoretic primitive named \emph{instantaneous modular flow}, which can be defined on generic multi-partite quantum states. It is useful not only in the proofs of the no-go theorems, but also has potential applications in other quantum information theory contexts. 

\emph{Definition.} Consider a multi-partite quantum state $\ket\psi$ on a tensor product Hilbert space $\calH = \otimes_i \calH_{A_i}$, where $\calH_{A_i}$ is the Hilbert space of a party $A_i$ of finite dimension.~\footnote{We consider this setup for simplicity and it is enough for the purpose of this paper. It is possible to generalize the definition for a system in a Hilbert space without tensor product structure or finite local dimension, such as continuum quantum field theory.} 
We will call a single party or a union of parties a \emph{region}, although it may not be associated with a spatial region. 
We use $\overline{A}$ to denote the complement of the region $A$.
\begin{definition}
    An instantaneous modular flow on a region $A$, denoted as $\ins_A(t)$, is a map 
    \begin{equation}
        \ins_A(t): \quad \ket{\psi} \mapsto \ins_A(t)\ket{\psi} = \rho_A^{\ii t}\ket{\psi}, \quad t\in\mathbb{R},
    \end{equation}
    where $\rho_A = \Tr_{\overline{A}}\ket{\psi}\bra{\psi}$ is the reduced density matrix\footnote{When $\rho_A$ has zero eigenvalues, we can set $0^{\ii t} \equiv 0$, because $\rho_A^{\ii t}$ is always acting on the state from which it is defined. Using singular value decomposition, one can see this is consistent with the property $\lim_{x \to 0} x^{\ii t} \sqrt{x} = 0$. Thus, using instantaneous modular flow avoids the issues arising from such zero eigenvalues.  See App.~\ref{app:non-invertible} for further discussion of this issue.} of $\ket{\psi}$ on region $A$. We call $\ket{\psi}$ the \emph{current state} of $\ins_A(t)$. 
\end{definition}
We emphasize that the actual unitary operator that implements an instantaneous modular flow is made from the reduced density matrix from the current state.
An instantaneous modular flow should be regarded as a nonlinear map that generates a certain state evolution. 
This is the main difference between instantaneous modular flows and (ordinary) modular flows defined in the existing literature (e.g.~\cite{pjm_1102810157,Lashkari:2018oke,Sorce:2023gio,Sorce:2024exl}).
For example, to compute $\ins_B(t)\ins_A(s)\ket{\psi}$, first do a modular flow $\ket{\psi'} = \ins_A(s)\ket{\psi} = \rho_A^{\ii s}\ket{\psi}$ using $\rho_A$ from $\ket{\psi}$. Then for $\ins_{B}(t)$, use $\rho'_B$ from the current state $\ket{\psi'}$ to obtain $(\rho'_B)^{\ii t}\ket{\psi'}$ as the final result, which is in general different from $\rho_B^{\ii t}\rho_A^{\ii s}\ket{\psi}$ with both $\rho_A,\rho_B$ from $\ket{\psi}$.
We remark that the focal point of the instantaneous modular flow is to consider their sequences and the resulting state evolutions, while the ordinary modular flow is usually regarded as a unitary evolution for operators in some algebra.

{\it Moves.} By definition, in a sequence of instantaneous modular flows, every modular flow always acts on the pure state from which the reduced density matrix for the flow is obtained. Because of this feature, instantaneous modular flows have many nice properties (called {\it moves}) that allow one to manipulate them in various ways and make them a valuable tool. Below we introduce two basic yet important moves.
(1) {\it Flipping move.} The most basic property is that 
\begin{equation}
    \ins_A(t)\ket{\psi} = \ins_{\overline{A}}(t)\ket{\psi},
\end{equation}
which is due to a simple fact that $\rho_A^{\ii t}\ket{\psi} = \rho_{\overline{A}}^{\ii t}\ket{\psi}$ with $\rho_A,\rho_{\overline{A}}$ from $\ket{\psi}$. Therefore, the manipulation $\ins_B(t)\ins_A(s)\ket{\psi} = \ins_{\overline{B}}(t)\ins_A(s)\ket{\psi}$ is allowed, while if one only used reduced density matrix $\rho_{A},\rho_B$ from $\ket{\psi}$ then $\rho_{B}^{\ii t}\rho_A^{\ii s}\ket{\psi} \neq \rho_{\overline{B}}^{\ii t}\rho_A^{\ii s}\ket{\psi}$ in general. 
(2) {\it Commutation move.} If $A \subseteq B$, then 
\begin{equation}
    \ins_{A}(s)\ins_B(t)\ket{\psi} = \ins_{B}(t)\ins_A(s)\ket{\psi}. 
\end{equation}
This is because on both hand side of the equation, one can use flip move $\ins_B(t) = \ins_{\overline{B}}(t)$, so that $\overline{B}$ does not overlap with $A$ and one can swap the two flows. In contrast, such a result with ordinary modular flows would not be true in general, $\rho_{A}^{\ii s}\rho_B^{\ii t}\ket{\psi} \neq \rho_B^{\ii t}\rho_A^{\ii s}\ket{\psi}$. One can also generalize this result to sequence of $n$ instantaneous modular flows $\ins_{A_n}(t_n)\cdots \ins_{A_2}(t_2)\ins_{A_1}(t_1)\ket{\psi}$. If $A_1\subseteq A_2 \subseteq \cdots\subseteq A_n$, then one can freely reorder the instantaneous modular flows in the sequence. 

\emph{Applications.} One can apply instantaneous modular flows to various contexts. Here we discuss its application in gapped quantum phases of matter and quantum information. More applications are mentioned in Sec.~\ref{sec:Discussion}. 

One can use instantaneous modular flow to compute topological quantities in gapped quantum phases of matter. For example, one can compute the modular commutator $J(A,B,C)_{\ket\Psi}$ via 
\begin{align}\label{eq:J-F}
    J(A,B,C)_{\ket\Psi} = 2 \Im \Big[\partial_x \partial_y \calF^{A,B,C}_{J}(x,y)_{\ket\Psi}\Big]_{x,y = 0},
\end{align}
where $\calF^{A,B,C}_{J}(x,y)_{\ket\Psi} \equiv \left(\ins_{AB}(x)\ket{\Psi}, \ins_{BC}(y)\ket{\Psi}\right)$.~\footnote{Here we use $(\cdot,\cdot)$ to denote state overlap: $( \ket\varphi, \ket\psi ) \equiv \langle \varphi|\psi\rangle$.} Similarly, one can compute $\Sigma(A,B,C)_{\ket\Psi}$ via 
\begin{equation}\label{eq:Sigma-F}
    \Sigma(A,B,C)_{\ket\Psi} = - \mathrm{Re} \Big[\partial_x \partial^2_y \calF^{A,B,C}_{\Sigma}(x,y)_{\ket\Psi}\Big]_{x,y = 0}  
\end{equation}
with $\calF^{A,B,C}_{\Sigma}(x,y)_{\ket\Psi} \equiv \left( \ins_{AB}(x)\ket{\Psi},U_{BC}(y)\ket\Psi \right)$.
This way of writing $J$ or $\Sigma$ is less singular in that it avoids using modular Hamiltonians; furthermore, the more general quantity $\CF$ turns out to be interesting in itself and more tractable than $J$ or $\Sigma$.

Another possible application of instantaneous modular flow is in quantum information theory. For example, one can obtain a necessary and sufficient condition for a state $\ket{\psi}$ in finite-dimensional Hilbert space to satisfy the Markov condition on $ABC$:
\begin{lemma}[Markov decomposition move]\label{lemma:Markov-flow}
    $I(A:C|B)_{\ket\psi}=0$ if and only if 
    \begin{equation}\label{eq:Markov-sequence}
      \ins_{ABC}(t)\ket{\psi} = \ins_{AB}(t)\ins_{BC}(t)\ins_B(-t)\ket{\psi},\; \forall t\in\mathbb{R}. 
    \end{equation}
\end{lemma}
We remark that in the sequence Eq.~\eqref{eq:Markov-sequence}, all the instantaneous modular flows commute. The proof of this lemma is given in App.~\ref{app:markov-flow}. We shall call the decomposition of $\ins_{ABC}(t)$ in Eq.~\eqref{eq:Markov-sequence} on a Markov state a \emph{Markov decomposition move}.

\section{Proofs of the no-go theorems}
The proof of Theorem~\ref{thm:no-go-thermal-intro} is by contradiction. If there were a state $\ket{\Psi}$ with the properties described in Theorem~\ref{thm:no-go-thermal-intro}, then based on $\ket{\Psi}$ one could construct a family of states $\{\rho_X(t)\}_{t\in\mathbb{R}}$ on a local Hilbert space $\calH_X = \otimes_{v \in X}\calH_v$ of a region $X$ (a collection of finite number of sites in the lattice), such that the entanglement entropies of these states $S(\rho_X(t))$ would be unbounded as a function of $t\in\mathbb{R}$. This contradicts the setup that $\calH_X$ has a finite dimension $d_X$, because entanglement entropies of any state within $\calH_X$ take values in range $[0,\ln d_X]$. 

The proof of Theorem~\ref{thm:no-go-electric-intro} follows a similar strategy. Roughly speaking, if there were a state $\ket{\Psi}$ with the properties described in Theorem~\ref{thm:no-go-electric-intro}, then one could increase the charge fluctuation of a local region infinitely via modular flow,
such that the charge fluctuation in the region would be infinite. This is in contradiction to the fact that $Q_v$ has a bounded spectrum for all $v$. 

Below we present the detailed proof of Theorem~\ref{thm:no-go-thermal-intro}, featured with usages of instantaneous modular flows. We leave the analogous proof of Theorem~\ref{thm:no-go-electric-intro} in App.~\ref{app:proof-no-go-electric}.

\begin{proof}[Proof of Theorem~\ref{thm:no-go-thermal-intro}]
The construction of such a one-parameter family of states is via a modular flow. Consider a tripartite region $ABC$ shown in Fig.~\ref{fig:bulk-A1-J} (Right), where each region $A,B,C$ contains sufficiently many\footnote{This is to make sure that it is possible to further divide $A,B,C$ into smaller regions.} but still a finite number of sites. Let 
\begin{equation}
    \ket{\Psi(t)} = \ins_{AB}(t)\ket{\Psi}.
\end{equation}
We wish to study the entanglement entropies $S_{BC}(t)\equiv S(\rho_{BC}(t))$ for a one-parameter family of states $\rho_{BC}(t) \equiv \Tr_{\overline{BC}}\ket{\Psi(t)}\bra{\Psi(t)}$, which are all within the same finite dimensional local Hilbert space $\calH_{BC}$. We shall show that if the original state $\ket{\Psi}$ satisfied bulk {\bf A1} and $J(A,B,C)_{\ket{\Psi}} = \frac{\pi c_{-}}{3} \neq 0$, then $S_{BC}(t) \equiv S(\rho_{BC}(t))$ would be unbounded as a function of $t\in\mathbb{R}$ and cause a contradiction.

To obtain a formula for $S_{BC}(t)$, we first compute its derivative for any $t$ and obtain a differential equation: 
\begin{equation}\label{eq:dSBC}\begin{split}
        \frac{d}{dt}S_{BC}(t)  &= \ii\bra{\Psi(t)}[K_{AB}(t),K_{BC}(t)]\ket{\Psi(t)} \\ &= J(A,B,C)_{\ket{\Psi(t)}},
\end{split}
\end{equation}
where $K_{AB}(t),K_{BC}(t)$ are modular Hamiltonians of $\rho_{AB}(t),\rho_{BC}(t)$ obtained from $\ket{\Psi(t)}$. Furthermore, as a result of Lemma~\ref{lemma:Jt} introduced below, we obtain 
\begin{equation}
    J(A,B,C)_{\ket{\Psi(t)}} = \frac{\pi c_{-}}{3},
\end{equation}
with which one can solve differential equation Eq.~\eqref{eq:dSBC}: 
\begin{equation}\label{eq:SBC-t}
    S_{BC}(t) = (S_{BC})_{\ket{\Psi}} + \frac{\pi c_{-}}{3}t,\quad \forall t\in\mathbb{R}.
\end{equation}
Therefore, $S_{BC}(t)$ is unbounded as a function of $t\in\mathbb{R}$ from both sides. This is a contradiction since $\rho_{BC}(t)$ in $\calH_{BC}$ of dimension $d_{BC}$ should have entanglement entropies within $[0,\ln d_{BC}]$.

\end{proof}

The key ingredient in the proof above is the following lemma: 
\begin{lemma}\label{lemma:Jt}
    If $\ket{\Psi}$ satisfies bulk {\bf A1}, for a sufficiently large tripartite region $ABC$ shown in Fig.~\ref{fig:bulk-A1-J} (Right)
    \begin{equation}\label{eq:Jt-c}
        J(A,B,C)_{\ket{\Psi(t)}} = J(A,B,C)_{\ket{\Psi}} = \frac{\pi c_{-}}{3},
    \end{equation}
    where $\ket{\Psi(t)} = \ins_{AB}(t)\ket{\Psi}$. 
\end{lemma}
{\it Proof sketch of Lemma~\ref{lemma:Jt}.} The insight behind this statement is the following two points: [{\bf P1}] With bulk {\bf A1}, one can do Markov decomposition moves for the instantaneous modular flows on $\ket{\Psi}$. [{\bf P2}] Instantaneous modular flows on $\ket{\Psi}$ preserve Markov conditions. Note that in showing Lemma~\ref{lemma:Jt} we only need a special case of {\bf P2}, which we will prove explicitly. The proof of the general case is deferred to \cite{IMF-bulk}.

\begin{figure}[!h]
    \centering
    \includegraphics[width=0.4\linewidth]{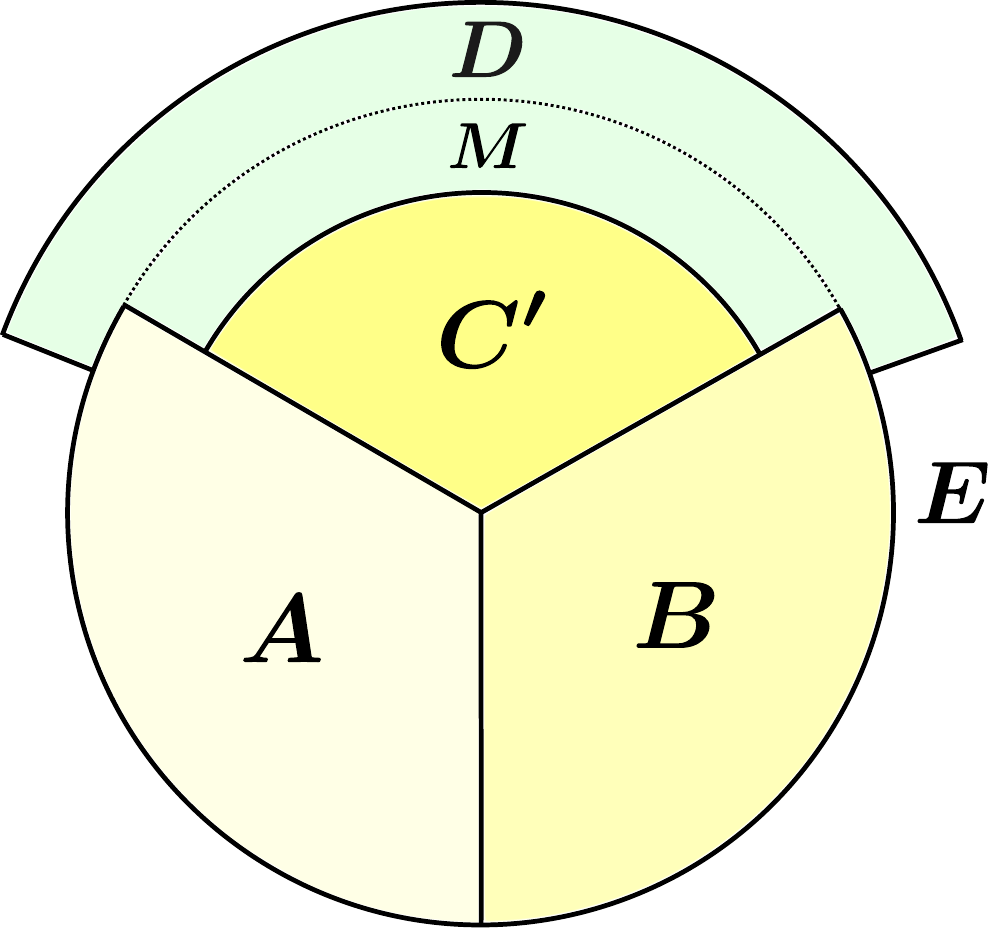}
    \caption{Regions for $J(A,B,C')_{\ket{\Psi'(t)}}$. $C = MC'$. Note that $E = \overline{ABCD}$ does not touch $C$.}
    \label{fig:deform}
\end{figure}
We now sketch the proof, whose details are in App.~\ref{app:lemma-Jt}. Consider a region $ABC$ shown in Fig.~\ref{fig:bulk-A1-J} (Right); the main task below is to show that 
\begin{equation}\label{eq:goal}
    J(A,B,C)_{\ket{\Psi(t)}} = J(A,B,C')_{\ket{\Psi'(t)}},
\end{equation}
where $\ket{\Psi'(t)} = \ins_{MD}(t)\ket{\Psi}$ and the regions $C',M,D$ are shown in Fig.~\ref{fig:deform}. Notice that the modular commutator on the right-hand side of Eq.~\eqref{eq:goal} only uses the reduced density matrix on $ABC'$, which is not changed by $\ins_{MD}$ [Fig.~\ref{fig:deform}]. With Eq.~\eqref{eq:goal}, we can conclude 
\begin{equation}
    J(A,B,C')_{\ket{\Psi'(t)}} = J(A,B,C')_{\ket{\Psi}} = \frac{\pi c_{-}}{3},
\end{equation}
and obtain Eq.~\eqref{eq:Jt-c}. 

The derivation of Eq.~\eqref{eq:goal} can be divided into two steps, which correspond to the two equalities:
\begin{equation}
\begin{split}
    J(A,B,C)_{\ket{\Psi(t)}}
    &\stackrel{\text{step 1}}{=}J(A,B,C)_{\ket{\Psi'(t)}} \\ 
    &\stackrel{\text{step 2}}{=}J(A,B,C')_{\ket{\Psi'(t)}}.
\end{split}
\end{equation}
{\it Step 1}: 
The goal of this step, which uses {\bf P1} above, is to show 
\begin{equation}\label{eq:F=F'}
    \calF_J^{A,B,C}(x,y)_{\ket{\Psi(t)}} = \calF_J^{A,B,C}(x,y)_{\ket{\Psi'(t)}},
\end{equation}
so that with Eq.~\eqref{eq:J-F}, one can obtain 
\begin{equation}
    J(A,B,C)_{\ket{\Psi(t)}} = J(A,B,C)_{\ket{\Psi'(t)}}. 
\end{equation}
In showing Eq.~\eqref{eq:F=F'}, we want to relate $\ins_{AB}(x)\ket{\Psi(t)}$ and $\ins_{BC}(y)\ket{\Psi(t)}$ to $\ins_{AB}(x)\ket{\Psi'(t)}$ and $\ins_{BC}(y)\ket{\Psi'(t)}$ by a unitary. Firstly, we can use the flipping move $\ket{\Psi(t)} = \ins_{AB}(t)\ket{\Psi} = \ins_{\overline{AB}}(t)\ket{\Psi}$. Secondly, with $\overline{AB}$ divided into small pieces as shown in Fig.~\ref{fig:deform}, one can use Markov decomposition moves to write $\ins_{\overline{AB}}(t)$ into a sequence of instantaneous modular flows supported on these small pieces.
This enables us to use the commutation move to pull $\ins_{AB}(x)$ and $\ins_{BC}(y)$ through some of the instantaneous modular flows in the sequence, obtaining 
\begin{equation}\label{eq:I-Vt}
    \begin{split}
        &\ins_{AB}(x)\ket{\Psi(t)} = \cdots \ins_{AB}(x)\ket{\Psi'(t)} = V(t)\ins_{AB}(x)\ket{\Psi'(t)} \\ 
        &\ins_{BC}(y)\ket{\Psi(t)} = \cdots \ins_{AB}(y)\ket{\Psi'(t)} = V(t)\ins_{BC}(y)\ket{\Psi'(t)}. 
    \end{split}
\end{equation}
$\cdots$ above denotes another sequence of instantaneous modular flows. One can further examine this sequence and show that its action on both $\ins_{AB}(x)\ket{\Psi'(t)}$ and $\ins_{BC}(y)\ket{\Psi'(t)}$ can be realized by the same unitary $V(t)$.  
Eq.~\eqref{eq:I-Vt} implies Eq.~\eqref{eq:F=F'} since the unitary $V(t)$ is canceled in the state overlap. 

{\it Step 2}: The second step is based on the insight {\bf P2}. The goal is to show 
\begin{equation}\label{eq:J=J'}
    J(A,B,C)_{\ket{\Psi'(t)}} = J(A,B,C')_{\ket{\Psi'(t)}}.  
\end{equation}
The deformation of the support of the modular commutator utilizes the Markov conditions in the same way as in \cite{Kim:2021tse,Kim2021}. Although the deformation in \cite{Kim:2021tse,Kim2021} is on the original state $\ket{\Psi}$, one can show that on the state $\ket{\Psi'(t)} = \ins_{MD}(t)\ket{\Psi}$, the Markov conditions required for the deformation in Eq.~\eqref{eq:J=J'} are preserved under $\ins_{MD}(t)$. 

In summary, combining Eq.~\eqref{eq:F=F'} and Eq.~\eqref{eq:J=J'}, we can conclude Eq.~\eqref{eq:goal} and finish the proof.  

\emph{Remarks:} 
Firstly, both no-go theorems are locally applicable. Even just on a local region\footnote{with sufficiently large, but still finite size} $\ket{\Psi}$ can not satisfy strict area law with $c_{-}\neq 0$ or $\sigma_{xy}\neq 0$. This is because the proofs only involve local entanglement properties of the quantum state: the strict area law is expressed with a local condition bulk {\bf A1}, and $c_{-},\sigma_{xy}$ can be computed also locally from the state. Secondly, in the setup of $\ket{\Psi}$, it is the finite local Hilbert dimension that contradicts to bulk {\bf A1} and $c_{-}\neq 0$ or $\sigma_{xy}\neq 0$, while the tensor product structure can be relaxed. For example, as discussed in App.~\ref{app:markov-flow}, the no-go theorems are applicable for fermionic systems, where the tensor product is $\mathbb{Z}_2$ graded.

\section{Discussion}
\label{sec:Discussion}
Our no-go theorems reveal that chirality, in the form of $c_{-}\neq 0$ or $\sigma_{xy}\neq 0$, constrains the local entanglement properties of a representative state of a chiral gapped phase\footnote{We comment that one might hope to combine the recent result of \cite{kim2024strict}, which says that any state satisfying {\bf A0} and {\bf A1} has a commuting projector parent Hamiltonian, with the Hamiltonian no-go theorems of \cite{Kapustin_2019,Zhang_2022,Kapustin_2020}, to provide a proof of a related result.}. 
Loosely speaking, our no-go theorems indicate that it is impossible to construct a state in a Hilbert space with finite local dimensions (or alternatively, with finite local entanglement entropies) that represent a zero correlation length fixed point of a chiral gapped phase.\footnote{This also aligns with various difficulties when people attempt to construct a chiral gapped state, such as using PEPS \cite{Hasik:2022,Wahl:2013rha} or MERA \cite{Chu:2023ven}. We leave the detailed connections to be investigated in the future.}
Moreover, the no-go theorems and their proofs might give a hint on studying the subleading correction terms to the area law entanglement formula Eq.~\eqref{eq:strict-area-law}. For a representative of a chiral gapped phase in a finite-dimensional local Hilbert space, we expect that the bulk {\bf A1} is approximately satisfied with violations decreasing as the subsystem sizes are enlarged, consistent with the decay of the subleading corrections to Eq.~\eqref{eq:strict-area-law}. On such a state, it is expected that $J$ and $\Sigma$ compute $c_{-}$ and $\sigma_{xy}$ also with decreasing error for larger subsystem sizes. One example of such a case is studied in the numerical section of \cite{emergence-of-virasoro2024}. In fact, the techniques used in the proofs might be used to relate the violation of the Markov flow equation Eq.~\eqref{eq:Markov-sequence}, which is a signature of the violation of bulk {\bf A1}, to the deviations of $J$ and $\Sigma$ from the expected values. The detailed relation is left for future investigation. 

At the end, we comment that the usage of the instantaneous modular flow in this paper is only the tip of an iceberg. One can apply it to study more properties of representative states of phases of matter, as well as critical systems such as conformal field theories (CFTs). For gapped phases of matter, one can, in fact, derive several flow equations, which are similar to Eq.~\eqref{eq:Markov-sequence} but contain universal quantities of the phase such as topological entanglement entropy and chiral central charge. 
For CFT, the instantaneous modular flows have a natural geometric picture, which further enables us to directly study the conformal group action. 
More detailed discussions of these applications will be in the upcoming work.

\begin{acknowledgments}
We are grateful to Isaac Kim for very useful discussions. 
This work was supported in part by
funds provided by the U.S. Department of Energy
(D.O.E.) under cooperative research agreement 
DE-SC0009919, 
and by the Simons Collaboration on Ultra-Quantum Matter, which is a grant from the Simons Foundation (652264, JM).
JM received travel reimbursement from the Simons Foundation;
the terms of this arrangement have been reviewed and approved by the University of California, San Diego in accordance with its conflict of interest policies.

\end{acknowledgments}

\appendix 
\clearpage
\renewcommand{\theequation}
{\Alph{section}.\arabic{equation}}

\section{Remarks on non-invertible density matrices}
\label{app:non-invertible}
In this section, we discuss instantaneous modular flows in circumstances where the reduced density matrices are not invertible. 

A reduced density matrix $\rho_X$ from a state $\ket{\psi}$ is not invertible when it contains zero eigenvalues. In the main text, we define the instantaneous modular flow $\ins_X(t)$ as 
\begin{equation}
    \ins_X(t)\ket{\psi} = \rho_X^{\ii t}\ket{\psi}. 
\end{equation}
When $\rho_X$ is not invertible, $\rho_X^{\ii t}$ should be understood via its spectral decomposition. Since $\rho_X$ is Hermitian, one can write it as 
\begin{equation}\label{eq:rho_X-decomp}
    \rho_X = \sum_a p_a \ket{a}\bra{a},
\end{equation}
with $p_a >0$. We then define 
\begin{equation}
    \rho_X^{\ii t}\equiv \sum_a p_a^{\ii t} \ket{a}\bra{a}. 
\end{equation}
Notice since $\rho_X$ is non-invertible, the eigenvectors $\ket{a}$ in Eq.~\eqref{eq:rho_X-decomp} do not span the whole Hilbert space $\calH_X$ and $\rho_X^{\ii t}$ is not a unitary operator on $\calH_X$. More explicitly, 
\begin{equation}\label{eq:rho-rho-inv}
    \rho_X^{\ii t}\rho_X^{-\ii t} = \rho_X^{-\ii t}\rho_X^{\ii t}=\Pi(\rho_X),
\end{equation}
where $\Pi(\rho_X)$ is a projector onto the support of $\rho_X$\footnote{$\rho_X^{-\ii t}$ is called a generalized inverse of $\rho_X^{\ii t}$, and vice versa \cite{penrose1955generalized,rao1972generalized}.}. 

In fact, $\rho_X^{\ii t}$ being non-unitary will not cause any trouble in the manipulations of instantaneous modular flows, because $\rho_X^{\ii t}$ always acts on its global state $\ket\psi$. Via singular value decomposition $\ket{\psi} = \sum_a \sqrt{p_a} \ket{a}_X \otimes \ket{a}_{\overline{X}}$ one can see
\begin{equation}\label{eq:Pi-action-1}
    \Pi(\rho_X) \ket{\psi} = \ket{\psi}. 
\end{equation}

More generally, for any state $\rho_Y$ with $X\subseteq Y$ such that $\rho_X = \Tr_{Y\setminus X}\rho_Y$, 
\begin{equation}\label{eq:Pi-action-2}
    \Pi(\rho_X) \rho_Y = \rho_Y.
\end{equation}
To see this, one can first purify $\rho_Y$ to $\ket{\varphi}$. Since $\ket{\varphi}$ gives reduced density matrix $\rho_X$, one can do singular value decomposition $\ket{\varphi} = \sum_a \sqrt{p_a} \ket{a}_X \otimes \ket{a}_{\overline{X}}$ and hence 
$\Pi(\rho_X) \ket{\varphi} = \ket{\varphi}$. Then with $\Pi(\rho_X) \ket{\varphi}\bra{\varphi} = \ket{\varphi}\bra{\varphi}$, tracing out the complement of $Y$, we obtain Eq.~\eqref{eq:Pi-action-2}. Similarly, one can also conclude 
\begin{equation}
    \rho_Y  \Pi(\rho_X)= \rho_Y. 
\end{equation}

Furthermore, even though $\rho_{X}^{\ii t}$ is not a unitary operator for non-invertible $\rho_X$, its action on the global state $\ket{\psi}$ can be realized by a unitary operator. One can always find a unitary $U_X(t)$ supported on $X$ such that 
\begin{equation}
    \ins_X(t)\ket{\psi} = \rho_X^{\ii t}\ket{\psi} = U_X(t)\ket{\psi}. 
\end{equation}
There are many choices of such unitaries. One simple choice is $\Tilde{\rho}_X^{\ii t}$ with 
\begin{equation}
    \Tilde{\rho}_X \equiv \sum_{a,p_a>0}p_a \ket{a}\bra{a} + \sum_{\alpha} \ket{\alpha}\bra{\alpha},
\end{equation}
where $\{\ket{\alpha}\}_{\alpha}$ together with $\{\ket{a}\}_a$ forms a complete set of basis vectors for $\calH_X$. One can easily verify that 
\begin{equation}
    \rho_X^{\ii t}\rho_Y = \Tilde{\rho}_X^{\ii t}\rho_Y,\quad \rho_Y\rho_X^{\ii t} = \rho_Y\Tilde{\rho}_X^{\ii t}.
\end{equation}
We call $\Tilde{\rho}^{\ii t}$ a {\it unitary extension} of $\rho_X^{\ii t}$. 

Later in the proofs we will remark on the non-invertible reduced density matrices in more details. 

\section{Markov decomposition move} 
\label{app:markov-flow}
In this appendix, we prove Lemma~\ref{lemma:Markov-flow}. 

\begin{proof}
To derive $I(A:C|B)_{\ket{\Psi}} = 0$ from the flow equation Eq.~\eqref{eq:Markov-sequence}, one can simply take the first derivative with $t$ at $t = 0$ on both hand side of Eq.~\eqref{eq:Markov-sequence}. The result is 
\begin{equation}\label{eq:K-markov-decomp-app}
    (K_{AB}+K_{BC}-K_{ABC} - K_B)\ket{\Psi} = 0,
\end{equation}
where the modular Hamiltonians are from $\ket\Psi$. Therefore, acting $\bra{\Psi}$ with it, since $\langle \Psi|K_X|\Psi\rangle = (S_X)_{\ket{\Psi}}$, one can obtain $I(A:C|B)_{\ket{\Psi}} = 0$. 

To derive the flow equation from $I(A:C|B)_{\ket{\Psi}} = 0$, we shall use the structure theorem of Markov states \cite{Hayden:2004wwj}: If $I(A:C|B)_{\ket{\Psi}} = 0$ on a tripartite region $ABC$, then there exists a decomposition of $\calH_B = \oplus_i \calH_{b_L^i}\otimes \calH_{b_R^i}$, such that $\rho_{ABC}$ from $\ket\Psi$ can be written as 
\begin{equation}\label{eq:rho-ABC-dec}
    \rho_{ABC} = \bigoplus_i p_i \rho_{Ab_L^i} \otimes \rho_{b_R^i C}. 
\end{equation}
With Eq.~\eqref{eq:rho-ABC-dec}, one can also write $\rho_{AB},\rho_{BC},\rho_B$ from $\ket\Psi$ as 
\begin{align}
    \rho_{AB} &= \bigoplus_i p_i \rho_{Ab_L^i} \otimes \rho_{b_R^i} \\ 
    \rho_{BC} &= \bigoplus_i p_i \rho_{b_L^i} \otimes \rho_{b_R^iC} \\ 
    \rho_{B} &= \bigoplus_i p_i \rho_{b_L^i} \otimes \rho_{b_R^i}. 
\end{align}

Notice in the decompositions above, we do not assume any of the reduced density matrices are invertible. We first show 
\begin{equation}\label{eq:markov-flow-rho}
    \rho_{AB}^{-\ii t}\rho_{ABC}^{\ii t} \cdot \rho_{ABC} = \rho_{B}^{-\ii t}\rho_{BC}^{\ii t}\cdot \rho_{ABC},\quad \forall t \in \mathbb{R}. 
\end{equation}
Utilizing the structure decompositions of $\rho_{ABC},\rho_{AB}$, we can obtain the left hand side of Eq.~\eqref{eq:markov-flow-rho}
\begin{equation}\label{eq:rhoAB-rhoABC}
    \begin{split}
        &\rho_{AB}^{-\ii t}\rho_{ABC}^{\ii t} \cdot \rho_{ABC} \\ 
        =&\bigoplus_i p_i \rho_{Ab_L^i}^{-\ii t} \rho_{Ab_L^i}^{\ii t}  \rho_{Ab_L^i} \otimes \rho_{b_R^i}^{-\ii t} \rho_{b_R^iC}^{\ii t}\rho_{b_R^iC} \\ 
        =&\bigoplus_i p_i  \rho_{Ab_L^i} \otimes \rho_{b_R^i}^{-\ii t} \rho_{b_R^iC}^{\ii t}\rho_{b_R^iC}, 
    \end{split}
\end{equation}
where to obtain the last line, we used the result $\rho_{Ab_L^i}^{-\ii t} \rho_{Ab_L^i}^{\ii t} \rho_{Ab_L^i} = \Pi(\rho_{Ab_L^i})\rho_{Ab_L^i} = \rho_{Ab_L^i}$. Similarly, one can also compute the right hand side of Eq.~\eqref{eq:markov-flow-rho} with the structure decompositions of $\rho_{ABC},\rho_{BC},\rho_B$: 
\begin{equation}\label{eq:rhoB-rhoBC}
    \begin{split}
        &\rho_{B}^{-\ii t}\rho_{BC}^{\ii t} \cdot \rho_{ABC} \\ 
        =&\bigoplus_i p_i \rho_{b_L^i}^{-\ii t} \rho_{b_L^i}^{\ii t}  \rho_{Ab_L^i} \otimes \rho_{b_R^i}^{-\ii t} \rho_{b_R^iC}^{\ii t}\rho_{b_R^iC} \\ 
        =&\bigoplus_i p_i  \rho_{Ab_L^i} \otimes \rho_{b_R^i}^{-\ii t} \rho_{b_R^iC}^{\ii t}\rho_{b_R^iC},
    \end{split}
\end{equation}
where the projector from $\rho_{b_L^i}^{-\ii t} \rho_{b_L^i}^{\ii t}$ was absorbed by $\rho_{Ab_L^i}$. Comparing the results in Eq.~\eqref{eq:rhoAB-rhoABC} and Eq.~\eqref{eq:rhoB-rhoBC}, we obtain Eq.~\eqref{eq:markov-flow-rho}, with which we can further conclude 
\begin{equation}\label{eq:markov-flow-psi}
    \rho_{AB}^{-\ii t}\rho_{ABC}^{\ii t}\ket{\Psi} = \rho_{B}^{-\ii t}\rho_{BC}^{\ii t}\ket{\Psi},
\end{equation}

Notice the left hand side of Eq.~\eqref{eq:markov-flow-psi} is $\ins_{AB}(-t)\ins_{ABC}(t)\ket{\Psi}$ and the right hand side is $\ins_{B}(t)\ins_{BC}(-t)\ket{\Psi}$, we can obtain 
\begin{equation}\label{eq:markov-app}
    \ins_{AB}(-t)\ins_{ABC}(t)\ket{\Psi} = \ins_{B}(-t)\ins_{BC}(t)\ket{\Psi},\quad \forall t\in\mathbb{R}. 
\end{equation}
Instantaneous modular flows have a property $\ins_X(-t)\ins_{X}(t)\ket{\psi} = \ket\psi$ for any $\ket\psi$ and $X$. Therefore, one can move the instantaneous modular flows on one side (say $\ins_{AB}(-t)$) of the equation \eqref{eq:markov-app} to the other side, and obtain 
\begin{equation}\label{eq:markov-app-2}
    \ins_{ABC}(t)\ket{\Psi}=\ins_{AB}(t)\ins_{BC}(t)\ins_{B}(-t)\ket{\Psi},\quad \forall t\in\mathbb{R}. 
\end{equation}

Moreover, one can see $\ins_{AB},\ins_{BC}$ also commute with each other. 
From Eq.~\eqref{eq:markov-app-2} by moving $\ins_B$ to the left-hand side, one obtain
\begin{equation}\label{eq:AB-BC}
    \ins_{ABC}(t)\ins_{B}(t)\ket{\psi} = \ins_{AB}(t)\ins_{BC}(t)\ket{\psi},\quad \forall t\in\mathbb{R}
\end{equation}
Then, by changing $t\to -t$ in Eq.~\eqref{eq:AB-BC} and move $\ins_{AB}(-t)\ins_{BC}(-t)$ to the left-hand side and $\ins_{ABC}(-t)\ins_B(-t)$ to the right-hand side, we obtain 
\begin{equation}
     \ins_{BC}(t)\ins_{AB}(t)\ket{\psi}  =  \ins_{ABC}(t)\ins_{B}(t)\ket{\psi},\quad \forall t\in\mathbb{R}.
\end{equation}
Together with Eq.~\eqref{eq:AB-BC}, we obtain
\begin{equation}
    \ins_{BC}(t)\ins_{AB}(t)\ket{\psi} =\ins_{AB}(t)\ins_{BC}(t)\ket{\psi},\quad \forall t\in\mathbb{R}.
\end{equation}

\end{proof}

{\it Remarks:} (1) The proof above is based on the structure theorem of Markov state, which is known to be applicable if the Hilbert space $\calH_{ABC}$ for $\rho_{ABC}$ is $\calH_{ABC} = \calH_A \otimes \calH_B \otimes \calH_C$ and has a finite dimension. We remark that the result (Theorem 2) in Ref.~\cite{Petz:2002eql} can also give a proof for Eq.~\eqref{eq:markov-flow-rho}. The proof in Ref.~\cite{Petz:2002eql} does not assume the density matrices are on a tensor product Hilbert space, and therefore we expect Lemma~\ref{lemma:Markov-flow} is also applicable in these circumstances. For example, we can consider a many-body system on a $\mathbb{Z}_2$ graded Hilbert space for fermionic degrees of freedom. With a proper definition of fermionic partial trace, such as in \cite{Vidal:2021xmv,Kim2021}, one can use Theorem 2 in \cite{Petz:2002eql} to conclude Eq.~\eqref{eq:markov-flow-rho} for such a Hilbert space. In more detail, Ref.~\cite{Kim2021} gives an explicit definition of fermionic partial trace in terms of Kraus representation, which is a completely-positive trace-preserving map as shown in Ref.~\cite{Vidal:2021xmv}. Then one can apply Theorem 2 in Ref.~\cite{Petz:2002eql} with such a map, to conclude Eq.~\eqref{eq:markov-flow-rho}. (2) We also remark that Ref.~\cite{Petz:2002eql} only considered the case when the reduced density matrices are invertible. We believe that one can modify the proof in \cite{Petz:2002eql} with a similar treatment of the non-invertible reduced density matrices given in App.~\ref{app:non-invertible}, and hence Eq.~\eqref{eq:markov-flow-psi} would still be valid. We leave this issue for future study.

\section{Properties of $J(A,B,C)_{\ket\Psi}$ and $\Sigma(A,B,C)_{\ket\Psi}$.}
\label{app:modular-commutators}
\subsection{Properties of modular commutators}
\label{app:mod-properties}
In this subsection, we summarize two important properties of modular commutators on $\ket{\Psi}$ in the setup introduced in the main text. The detailed proof of these properties is given in \cite{Kim:2021tse,Kim2021}. 

The first one is for generic tripartite regions $ABC$. 
\begin{Fact}\label{fact:J-markov}
    If $\ket{\Psi}$ satisfies $I(A:C|B)_{\ket\Psi}=0$, then 
    \begin{equation}
        J(A,B,C)_{\ket\Psi} = 0. 
    \end{equation}
\end{Fact}
Notice in the statement \ref{fact:J-markov}, there is no specific requirement of the topology of the tripartite region $ABC$. 

\begin{Fact}\label{fact:J-def-inv}
    If $\ket\Psi$ satisfies bulk {\bf A1}, then $J(A,B,C)_{\ket{\Psi}}$ is a constant for any choice of $ABC$ in the bulk which is topologically equivalent to the one shown in Fig.~\ref{fig:bulk-A1-J} (Right).  
\end{Fact}
This fact indicates that $J(A,B,C)_{\ket\Psi}$ for $ABC$ chosen as Fig.~\ref{fig:bulk-A1-J} (Right) is a topological quantity. Moreover, it is verified numerically that it indeed computes chiral central charge. A theoretical argument of why it computes chiral central charge in a representative state of a chiral gapped phase of matter will be given in \cite{Corner_ent}.

\subsection{Properties of electric Hall conductance}
\label{app:sigma-property}
In this subsection, we introduce two important properties of $\Sigma(A,B,C)_{\ket\Psi}$, which are in analog of the properties for $J(A,B,C)_{\ket\Psi}$ introduced above. One is for the quantum Markov state, and the other is the deformation invariance (i.e., topological) property. 

In particular, we shall provide a proof of the topological property (i.e. deformation invariance) of $\Sigma(A,B,C)_{\ket\Psi}$, which is different from the one given in \cite{fan2022extracting}. Our proof only utilizes bulk {\bf A1} and avoids the ``cluster property'' used in \cite{fan2022extracting}. 

In the following statements and the proofs, we shall use $\ket\Psi$ to denote the quantum many-body state with on-site $\gU(1)$ symmetry under the setup introduced in the main text and $K_{X}$ to denote the modular Hamiltonian of $\rho_X$ from $\ket{\Psi}$. We shall also use the notation 
\begin{equation}
    \langle O \rangle \equiv \langle \Psi| O |\Psi\rangle,
\end{equation}
for the expectation value of an operator $O$ on $\ket{\Psi}$. We  

\subsubsection{Basic facts about symmetry charge operators}
Before we start to discuss the properties, it is useful to first list some of the basic facts of the symmetry operator. We shall heavily use these facts in the later discussions.  

\begin{Fact}[Basic facts] \label{fact:basic}
    Below we collect four basic facts about the local charge operators of the on-site $\gU(1)$ symmetry and modular Hamiltonians from the symmetric state:  
    \begin{enumerate}
        \item $Q_{AB} = Q_A + Q_B$ for any regions $A,B$.  
        \item $[K_X,Q_X]=0$. Moreover, $[K_A,Q_{AB}] = 0$ and $ [K_{A},Q_{AB}^2] = 0$ for any regions $A,B$. 
        \item $\langle [K_{AB},Q_A^2]\rangle = \langle [K_{A},Q_{AB}^2]\rangle = 0$ for any region $AB$. 
        \item $\langle [K_{AB},Q_{BC}^2] \rangle = \langle [K_{AB},Q_{\overline{BC}}^2] \rangle$ for any tripartite region $ABC$. 
    \end{enumerate}
\end{Fact}
\begin{proof}
    Fact (1) is simply because the local charge operator is a sum over on-site operators. That is, for a region $X$, $Q_X = \sum_{v\in X} Q_v$. 

    Fact (2) is the consequence of the symmetry. Explicitly, $U(t)\ket{\Psi} = U_X(t)U_{\overline{X}}(t)\ket{\Psi} = \alpha \ket\Psi$, where $\alpha$ is a phase. Then $U_{X}(t)\ket{\Psi} = \alpha U_{\overline{X}}^{\dagger}(t)\ket{\Psi}$ and hence 
    \begin{equation}\label{eq:U-rho-commuting}\begin{split}
        &U_{X}(t)\ket{\Psi} \bra{\Psi}U_{X}(t)^{\dagger} = U_{\overline{X}}(t)^{\dagger}\ket{\Psi} \bra{\Psi}U_{\overline{X}}(t) \\ 
        \Rightarrow &\quad U_X(t)\rho_X U_X(t)^{\dagger} = \rho_X,
    \end{split}
    \end{equation}
    where the $\Rightarrow$ can be seen by tracing out $\overline{X}$. Then one can see $U_X(t)K_X U_X(t)^{\dagger} = K_X$ since $K_X = -\ln \rho_X$. Taking the first derivative of $t$ at $t=0$, we obtain $[K_X,Q_X]=0$. With $Q_{AB} = Q_{A}+Q_B$, we can see $[K_A,Q_{AB}]=0$. To see $[K_{A},Q_{AB}^2]= 0$, we can utilize the Fact (1) to write $Q_{AB} = Q_A + Q_B$. Therefore, 
    \begin{equation}\begin{split}
        & [K_{A},Q_{AB}^2]\\ 
        =& [K_{A},Q_{A}^2] +  [K_{A},Q_{B}^2] + 2  [K_{A},Q_{A}Q_B] \\ 
        =&0.
    \end{split}
    \end{equation}
    The three commutators above all vanish because $[K_A,Q_A]=0$ and $[K_A,Q_B]=0$.  
        
    Now we explain Fact (3). To see $\langle [K_{AB},Q_A^2]\rangle = 0$, we utilize the flipping $K_{AB}\ket\Psi = K_{\overline{AB}}\ket{\Psi}$ due to the property of pure state. Therefore, 
    \begin{equation}
        \langle [K_{AB},Q_A^2]\rangle = \langle [K_{\overline{AB}},Q_A^2]\rangle = 0,
    \end{equation}
    as the compliment of $AB$ has no overlap with $A$. 

    Fact (4) is for any tripartite region $ABC$. Because the state has $\gU(1)$ symmetry, $\ket\Psi$ is an eigenstate of the symmetry operator $Q$. Let $q$ be the eigenvalue, we thus have $Q\ket{\Psi} = (Q_{BC}+Q_{\overline{BC}})\ket{\Psi} = q \ket\Psi$. Therefore, we can replace $Q_{BC}\ket\Psi$ with $(q-Q_{\overline{BC}})\ket\Psi$ and as a result 
    \begin{equation}\begin{split}
        \langle [K_{AB},Q_{BC}^2] \rangle &= \langle [K_{AB},q^2-2q Q_{\overline{BC}} + Q_{\overline{BC}}^2]\rangle \\ 
        &= \langle [K_{AB}, Q_{\overline{BC}}^2]\rangle. 
    \end{split}
    \end{equation}
    In the derivation of the last line in the equation above, $\langle [K_{AB},q^2]\rangle = 0$ is obvious as $q$ is simply a number. To see $\langle [K_{AB},2q Q_{BC}]\rangle = 0$, we can show that 
    \begin{equation}\begin{split}
        \langle [K_{AB},Q_{BC}]\rangle =& \langle [K_{AB},Q_{B}]\rangle + \langle [K_{AB},Q_{C}]\rangle  \\ 
        =& \langle [K_{\overline{AB}},Q_{B}]\rangle = 0. 
    \end{split}
    \end{equation} 
\end{proof}

\subsubsection{Properties of $\Sigma(A,B,C)$}

Now we discuss two important properties of $\Sigma(A,B,C)$. We remind readers that 
\begin{equation}
    \Sigma(A,B,C)_{\ket\Psi} = \frac{\ii}{2}\langle [K_{AB},Q_{BC}^2]\rangle,
\end{equation}
therefore, we shall mainly work with the commutator $\langle [K_{AB},Q_{BC}^2]\rangle$ without bothering to convert it to $\Sigma(A,B,C)$. 

\begin{lemma}\label{lemma:Sigma-markov}
    If $I(A:C|B)_{\ket\Psi} = 0$, then 
    \begin{equation} \label{eq:Sigma-markov-app-0}
       \langle [K_{AB},Q_{BC}^2]\rangle = 2\langle [K_{AB},Q_B Q_C] \rangle = 0.
    \end{equation}
\end{lemma}
This implies $\Sigma(A,B,C)_{\ket{\Psi}}=0$ if $\ket{\Psi}$ satisfies Markov condition on $ABC$. 
\begin{proof}
    The first equality in Eq.~\eqref{eq:Sigma-markov-app-0} can be obtained by $Q_{BC} = Q_B + Q_C$ [Fact~\ref{fact:basic} (1)]: 
    \begin{equation}\label{eq:Sigma-markov-app}
        \begin{split}
        &\langle [K_{AB},Q_{BC}^2] \rangle \\ =& \langle [K_{AB},Q_{B}^2] \rangle + 2 \langle [K_{AB},Q_B Q_C] \rangle + \langle [K_{AB},Q_{C}^2] \rangle \\ 
        =&2 \langle [K_{AB},Q_B Q_C] \rangle, 
    \end{split}
    \end{equation}
    where first vanishes because of [Fact~\ref{fact:basic} (3)] and the third term in Eq.~\eqref{eq:Sigma-markov-app} vanishes since the operators in the commutator has no common support. 

Now we shall show the remaining term in Eq.~\eqref{eq:Sigma-markov-app} also vanishes. The key is the Markov condition. With $I(A:C|B)_{\ket\Psi} =0$, one can write 
\begin{equation}\label{eq:decompose-KABC-app}
    K_{ABC}\ket\Psi = (K_{AB}+K_{BC}-K_B)\ket\Psi .
\end{equation}
Therefore 
\begin{equation} \begin{split}
    &\langle [K_{AB},Q_B Q_C] \rangle \\ 
    =& \langle [ K_{ABC}, Q_B Q_C] \rangle - \langle [ K_{BC}, Q_B Q_C] \rangle + \langle [ K_{B}, Q_B Q_C] \rangle \\ 
    =&\langle [ K_{\overline{ABC}}, Q_B Q_C] \rangle - \langle [ K_{\overline{BC}}, Q_B Q_C] \rangle + \langle [ K_{B}, Q_B Q_C] \rangle \\ 
    =&\langle [ K_{B}, Q_B Q_C]\rangle \\ 
    =& \langle Q_B[ K_{B},  Q_C] \rangle + \langle [ K_{B},  Q_B] Q_C \rangle \\ 
    =&0.  
\end{split}
\end{equation}
In the derivation, we first utilize Eq.~\eqref{eq:decompose-KABC-app} to replace $K_{AB}\ket\Psi$. Then, in the second line to the third line, we flip the support of $K_{ABC},K_{BC}$ to their compliments, which have no overlap with $Q_B Q_C$ and hence the first two commutators in the third line vanish. To show the remaining term vanishes $\langle [ K_{B}, Q_B Q_C] \rangle = 0$, we utilize $[K_B,Q_B]=  0$ [Fact~\ref{fact:basic} (2)] in the second last line. 
\end{proof}

Now we are ready to discuss the topological property. 
\begin{lemma}\label{lemma:topology-Sigma}
    If $\ket\Psi$ satisfies bulk {\bf A1}, then $\Sigma(A,B,C)_{\ket{\Psi}}$ is a constant for any choice of $ABC$ topologically equivalent to the one shown in Fig.~\ref{fig:bulk-A1-J} (Right). 
\end{lemma}
To prove this property, we need to show that $\Sigma(A,B,C)$ is invariant under deformations of $ABC$ that does not change the topology and orientation of the tripartite region. We shall heavily use Fact~\ref{fact:markov}, which is the consequence of the bulk {\bf A1}. 

\begin{figure}[!h]
    \centering
    \includegraphics[width=0.8\linewidth]{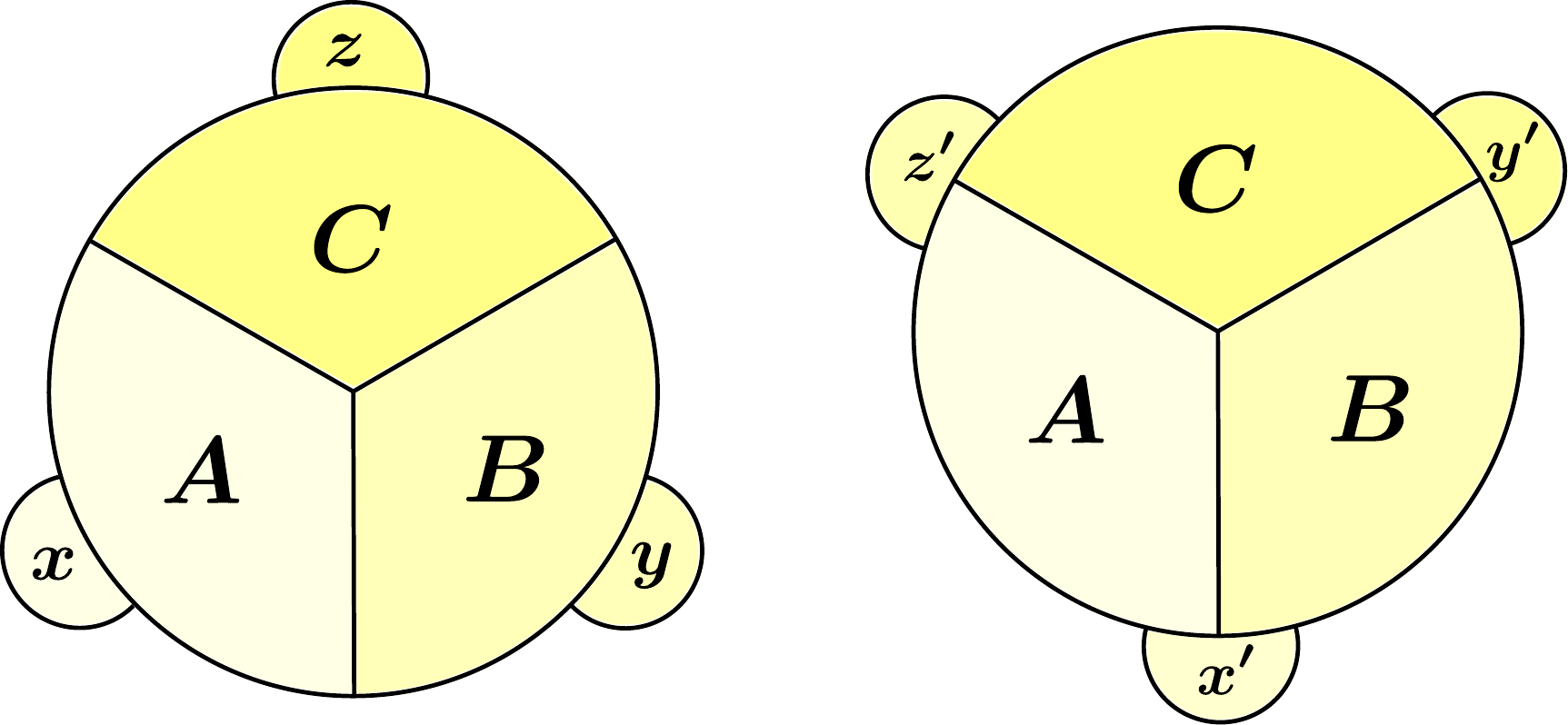}
    \caption{Various region deformations of $\Sigma(A,B,C)$ away from (left) and near (right) a triple intersection point among $A,B,C,\overline{ABC}$.}
    \label{fig:Sigma-deform}
\end{figure}

We first study the deformations that are away from the triple intersection points, shown in Fig.~\ref{fig:Sigma-deform} (Left). The derivation utilizes Markov conditions $I(x:B|A)=0,I(y:A|B)=0,I(z:A|B)=0$ from Fact~\ref{fact:markov}\footnote{The derivation of the last Markov condition from Fact~\ref{fact:markov} is explained later.}. 
\begin{itemize}
    \item {\it (Case 1)} $A \to Ax$: We wish to show 
    \begin{equation}\label{eq:case1}
        \langle [K_{ABx},Q_{BC}^2]\rangle = \langle [K_{AB},Q_{BC}^2]\rangle. 
    \end{equation}
    The derivation uses $I(x:B|A) = 0$ [Fact~\ref{fact:markov}]. With this Markov condition, we can substitude $K_{ABx}\ket\Psi$ with 
    \begin{equation}
        K_{ABx}\ket\Psi = (K_{AB}+K_{Ax}-K_A)\ket\Psi, 
    \end{equation}
    on the left-hand side of Eq.~\eqref{eq:case1}. Then the terms that involve $K_{Ax},K_{A}$ vanish because they have no overlapping support with $Q_{BC}^2$ and we obtain Eq.~\eqref{eq:case1}. 

    \item {\it (Case 2)} $B \to By$: We wish to show
    \begin{equation}\label{eq:case2}
        \langle [ K_{ABy}, Q_{BCy}^2]\rangle = \langle [ K_{AB}, Q_{BC}^2]\rangle. 
    \end{equation}
    The derivation uses $I(y:A|B)=0$ [Fact~\ref{fact:markov}]. With this Markov condition, we can make the substitution of $K_{ABy}\ket\Psi$ with 
    \begin{equation}
        K_{ABy}\ket\Psi = (K_{AB}+K_{By}-K_B)\ket\Psi
    \end{equation}
    on the left-hand side of Eq.~\eqref{eq:case2}, and obtain 
    \begin{equation}\label{eq:case2-derive-1}
        \begin{split}
            \langle [ K_{ABy}, Q_{BCy}^2]\rangle =&\langle [ K_{AB}, Q_{BCy}^2]\rangle \\ +& \langle [ K_{By}, Q_{BCy}^2]\rangle - \langle [ K_{B}, Q_{BCy}^2]\rangle \\ 
            =& \langle [ K_{AB}, Q_{BCy}^2]\rangle. 
        \end{split}
    \end{equation}
    The vanishing of the last two terms on the right-hand side of Eq.~\eqref{eq:case2-derive-1} is because of [Fact~\ref{fact:basic} (2)]. Then we can use $Q_{BCy} = Q_{B} + Q_C + Q_y$ to obtain 
    \begin{equation}\label{eq:case2-derive-2}
    \begin{split}
        \langle [ K_{AB}, Q_{BCy}^2]\rangle &=\langle [ K_{AB}, Q_{BC}^2]\rangle + 2 \langle[K_{AB},Q_B Q_y]\rangle  \\ 
        & = \langle [ K_{AB}, Q_{BC}^2]\rangle, 
    \end{split}
    \end{equation}
    where $\langle [ K_{AB}, Q_{B}Q_y]\rangle = 0$, since $I(y:A|B)=0$ with Lemma~\ref{lemma:Sigma-markov}.
    \item {\it (Case 3)} $C \to Cz$: We wish to show
    \begin{equation}\label{eq:case3}
        \langle [K_{AB},Q_{BCz}^2]\rangle =\langle [K_{AB},Q_{BC}^2]\rangle.  
    \end{equation}
    The derivation uses $I(z:A|B)=0$. To see this condition is satisfied, one can first consider a disk-like region $Z$ that contains $z$ and touches $B$ while does not touch $A$. Then $I(A:Z|B)=0$ due to [Fact~\ref{fact:markov}]. Therefore, 
    \begin{equation}
        \begin{split}
            0\leq I(z:A|B) \leq I(Z:A|B)=0,
        \end{split}
    \end{equation}
    where the two inequalities are due to strong subadditivity (SSA) \cite{Petz:2002eql}. Therefore  $I(z:A|B)=0$. With this condition and $Q_{BCz} = Q_{B}+Q_C+Q_z$, we can see 
    \begin{equation}
    \begin{split}
        \langle [K_{AB},Q_{BCz}^2]\rangle &= \langle [K_{AB},Q_{BC}^2]\rangle + 2\langle [K_{AB},Q_{B}Q_z]\rangle \\ &=\langle [K_{AB},Q_{BC}^2]\rangle,
    \end{split}
    \end{equation}
    where $\langle [K_{AB},Q_{B}Q_z]\rangle=0$ is due to $I(z:A|B)=0$ and [Lemma~\ref{lemma:Sigma-markov}]. 
\end{itemize}
Now we study the deformations that touch the triple intersection points, shown in Fig.~\ref{fig:Sigma-deform} (Right). We shall only focus on the deformations $A\to Ax'$, $B\to Bx'$ and $C\to Cz'$. There are three other possible cases in the setup of Fig.~\ref{fig:Sigma-deform} (Right), namely $A\to Az'$, $B\to By'$ and $C\to Cy'$. They are in fact the same as the case (1-3), because one can simply repeat the proof using $I(z':B|A)=0,I(y':A|B)=0,I(y':A|B)=0$ as in case 1, 2, 3 respectively. 
\begin{itemize}
    \item {\it (Case 4)} $A\to Ax'$: We wish to show 
    \begin{equation}\label{eq:case4}
        \langle [K_{ABx'},Q_{BC}^2] \rangle = \langle [K_{AB},Q_{BC}^2] \rangle. 
    \end{equation}
    We can use $I(x':C|AB)=0$ to write 
    \begin{equation}\label{eq:KABx'-decomp}
        K_{ABx'}\ket{\Psi} = (K_{AB}+K_{ABCx'}-K_{ABC})\ket{\Psi}.
    \end{equation}
    Plug this into the left-hand side of Eq.~\eqref{eq:case4}, we obtain 
    \begin{equation}\label{eq:case4-derive}
        \begin{split}
            &\langle [K_{ABx'}, Q_{BC}^2]\rangle \\ 
            =&\langle [K_{AB}, Q_{BC}^2]\rangle \\
            +&\langle [K_{ABCx'}, Q_{BC}^2]\rangle - \langle [K_{ABC}, Q_{BC}^2]\rangle \\ 
            =&\langle [K_{AB}, Q_{BC}^2]\rangle. 
        \end{split}
    \end{equation}
    The terms in the third line of Eq.~\eqref{eq:case4-derive} vanish because of Fact~\ref{fact:basic} (3).
    
    \item {\it (Case 5)} $B\to Bx'$: We wish to show 
    \begin{equation}\label{eq:case5}
        \langle [K_{ABx'},Q_{BCx'}^2]\rangle = \langle [K_{AB},Q_{BC}^2]\rangle. 
    \end{equation}
    With $Q_{BCx'} = Q_{B}+Q_C+Q_{x'}$, we obtain 
    \begin{equation}\label{eq:case5-derive}
    \begin{split}
        \langle [K_{ABx'},Q_{BCx'}^2]\rangle &= \langle [K_{ABx'},Q_{BC}^2]\rangle \\
        &+ 2\langle [K_{ABx'},Q_{B}Q_{x'}]\rangle \\ 
        & = \langle [K_{ABx'},Q_{BC}^2]\rangle,
    \end{split}
    \end{equation}
    where the second term on the right-hand side of Eq.~\eqref{eq:case5-derive} vanishes because 
    \begin{equation}
        \begin{split}
        \langle [K_{ABx'},Q_{BC}Q_{x'}]\rangle &=\langle [K_{\overline{ABx'}},Q_{B}Q_{x'}]\rangle=0.
        \end{split}
    \end{equation}
    With the result of Eq.~\eqref{eq:case5-derive}, together with the result in {\it case 4}, we derived Eq.~\eqref{eq:case5}
    
    \item {\it (Case 6)} $C\to Cz'$: Firstly, with $Q_{BCz'} = Q_{B}+Q_C + Q_{z'}$, we can obtain 
    \begin{equation}\label{eq:Sigma-AB-BCz}\begin{split}
        \langle [K_{AB},Q_{BCz}^2] \rangle =& \langle [K_{AB},Q_{BC}^2] \rangle + 2 \langle [K_{AB},Q_{B}Q_{z'}] \rangle \\ 
        =& \langle [K_{AB},Q_{BC}^2] \rangle. 
    \end{split}
    \end{equation}
    We shall show $\langle [K_{AB},Q_{B}Q_{z'}] \rangle$ in Eq.~\eqref{eq:Sigma-AB-BCz} vanishes as follows. Firstly, we can obtain 
    \begin{equation}
        \begin{split}
            \langle [K_{AB},Q_{B}Q_{z'}] \rangle = -\langle [K_{AB},Q_{A}Q_{z'}] \rangle,
        \end{split}
    \end{equation}
    because $[K_{AB},Q_A Q_{z'}]+[K_{AB},Q_{B}Q_{z'}]=[K_{AB},Q_{AB}Q_{z'}] =0$. Then, notice $I(z':B|A)=0$, which enable us to replace $K_{AB}$ in the commutator using 
    \begin{equation}
        K_{AB}\ket\Psi = (K_{ABz'}+K_{A}-K_{Az'})\ket\Psi. 
    \end{equation}
    As a result, we obtain 
    \begin{equation}\begin{split}
        -&\langle [K_{AB},Q_{A}Q_{z'}] \rangle = \langle [K_{Az'},Q_{A}Q_{z'}] \rangle \\ 
        -& \langle [K_{ABz'},Q_{A}Q_{z'}] \rangle - \langle [K_{A},Q_{A}Q_{z'}] \rangle \\ 
        =&0.
    \end{split}        
    \end{equation}
    In the equation above, $\langle [K_{Az'},Q_{A}Q_{z'}] \rangle = \langle [K_{ABz'},Q_{A}Q_{z'}] \rangle=0$ because one can flip the support of $K_{Az'}$ and $K_{ABz'}$ to their compliments so that they have no overlap with $Az'$. $\langle [K_{A},Q_{A}Q_{z'}] \rangle = 0$ is simply because $[K_A,Q_A Q_{z'}] = [K_A,Q_A]Q_{z'}=0$ with [Fact~\ref{fact:basic} (2)]. 
\end{itemize}

By now, we've considered all the possible cases of deformations of the region $ABC$ that happen at the boundary of $ABC$. One can also deform the other boundaries of $A$, $B$ and $C$ that are inside the tripartite region $ABC$. This is because one can flip the support of $K_{AB}$ or $Q_{BC}^2$ to its compliment using the pure state property and Fact~\ref{fact:basic} (4) respectively, so that the deformation is no longer inside the flipped tripartite region and therefore is within the cases (1-6) considered before.

\section{Proof of Lemma~\ref{lemma:Jt}}
\label{app:lemma-Jt}

\begin{proof}
    As we explained in the main text, to prove 
    \begin{equation}
        J(A,B,C)_{\ket\Psi} = J(A,B,C)_{\ket{\Psi(t)}},
    \end{equation}
    where $\ket{\Psi(t)} = \ins_{AB}(t)\ket\Psi$, the key is to show 
    \begin{equation}\label{eq:J-J'-app}
        J(A,B,C)_{\ket{\Psi(t)}} = J(A,B,C')_{\ket{\Psi'(t)}},
    \end{equation}
    where $\ket{\Psi'(t)} = \ins_{MD}(t)\ket\Psi$ and $C = MC'$ shown in Fig.~\ref{fig:deform}. With Eq.~\eqref{eq:J-J'-app} derived, since the reduced density matrix of $\ket{\Psi'(t)}$ on $ABC'$ is the same as the one of $\ket\Psi$, we can obtain $J(A,B,C')_{\ket{\Psi'(t)}} = J(A,B,C')_{\ket{\Psi}}$, which is equal to $J(A,B,C)_{\ket{\Psi}}$ by its deformation invariant property [Fact~\ref{fact:J-def-inv}]. 

    We now begin the proof. 

    {\it (Step 1)}: The first step is to show 
    \begin{equation}\label{eq:J-J-psi'-app}
        J(A,B,C)_{\ket{\Psi(t)}} = J(A,B,C)_{\ket{\Psi'(t)}}.
    \end{equation}
    This is a consequence of Eq.~\eqref{eq:F=F'}, which we copy here:
    \begin{equation}\label{eq:F-F'-app}
        \calF^{A,B,C}_J(x,y)_{\ket{\Psi(t)}} = \calF^{A,B,C}_J(x,y)_{\ket{\Psi'(t)}},
    \end{equation}
    where $\calF^{A,B,C}_J(x,y)$ is explicitly defined below Eq.~\eqref{eq:J-F}. The left-hand side of Eq.~\eqref{eq:F-F'-app} involves $\ins_{AB}(x)\ket{\Psi(t)}$ and $\ins_{BC}(y)\ket{\Psi(t)}$, and we wish to show that they are related to $\ins_{AB}(x)\ket{\Psi'(t)}$ and $\ins_{BC}(y)\ket{\Psi'(t)}$ by the same unitary. 

    To show this, we first make connection between $\ket{\Psi(t)} = \ins_{AB}(t)\ket\Psi$ and $\ket{\Psi'(t)} = \ins_{MD}(t)\ket{\Psi}$ via the following two Markov decomposition moves: Since $I(C:E|D)_{\ket{\Psi}}=0$ [Fact~\ref{fact:markov} and Lemma~\ref{lemma:Markov-flow}], one can obtain
    \begin{equation}\label{eq:IAB-decompose}\begin{split}
        \ins_{AB}(t)\ket\Psi &= \ins_{CDE}(t)\ket\Psi \\ 
        &=\ins_{D}(-t)\ins_{DE}(t)\ins_{CD}(t)\ket{\Psi}
    \end{split}    
    \end{equation}
    Furthermore, for $\ins_{CD}(t)\ket{\Psi}$ in the result above, we can do another Markov decomposition move because of $I(D:C'|M)_{\ket\Psi}=0$:
    \begin{equation}
        \begin{split}
            \ins_{CD}(t)\ket{\Psi} &= \ins_{M}(-t)\ins_{C}(t)\ins_{MD}(t)\ket{\Psi} \\ 
            &= \ins_{M}(-t)\ins_{C}(t)\ket{\Psi'(t)}
        \end{split}
    \end{equation} 
    Combining the results of the two Markov decomposition moves, we obtain 
    \begin{equation}
        \label{eq:Psi-t-deform}
        \ket{\Psi(t)} = \ins_{D}(-t)\ins_{DE}(t) \ins_{M}(-t)\ins_{C}(t)\ket{\Psi'(t)}.
    \end{equation}

    With Eq.~\eqref{eq:Psi-t-deform} plugged in $\ins_{AB}(x)\ket{\Psi(t)}$ and $\ins_{BC}(y)\ket{\Psi(t)}$, we can obtain
    \begin{equation}\label{eq:IABx}
    \begin{split}
        &\ins_{AB}(x)\ket{\Psi(t)}  \\ 
        =& \ins_{AB}(x)\ins_{DE}(t) \ins_{D}(-t)\ins_{C}(t)\ins_{M}(-t)\ket{\Psi'(t)} \\ 
        =& \ins_{D}(-t) \ins_{DE}(t)\ins_{M}(-t)\ins_{C}(t)\ins_{AB}(x)\ket{\Psi'(t)}
    \end{split}
    \end{equation}
    and 
    \begin{equation}\label{eq:IBCy}
        \begin{split}
            &\ins_{BC}(y)\ket{\Psi(t)} \\ =&  \ins_{BC}(y)\ins_{DE}(t) \ins_{D}(-t)\ins_{C}(t)\ins_{M}(-t)\ket{\Psi'(t)} \\ 
            =&\ins_{D}(-t) \ins_{DE}(t)\ins_{M}(-t)\ins_{C}(t)\ins_{BC}(y)\ket{\Psi'(t)}. 
    \end{split}
    \end{equation}
    In the derivations above, we utilized the commutation move that we introduced in Sec.~\ref{sec:IMF}. 

    Moreover, we can show that the instantaneous modular flow sequence $\ins_{DE}(t) \ins_{D}(-t)\ins_{C}(t)\ins_{M}(-t)$ in both Eq.~\eqref{eq:IABx} and Eq.~\eqref{eq:IBCy} can be replaced by the same unitary operator $V(t)$, so that
    \begin{align}
        \ins_{AB}(x)\ket{\Psi(t)} &= V(t)\ins_{AB}(x)\ket{\Psi'(t)} \label{eq:IABx-Vt}\\ 
        \ins_{BC}(y)\ket{\Psi(t)} &= V(t)\ins_{BC}(y)\ket{\Psi'(t)}. 
    \end{align}
    The unitary $V(t)$ is defined as 
    \begin{equation}\label{eq:def-Vt}
        V(t) \equiv (\rho'_{D})^{-\ii t} (\rho'_{DE})^{\ii t} (\rho'_{M})^{-\ii t} (\rho'_{C})^{\ii t},
    \end{equation}
    where $\rho'_{\bullet}$ denotes the reduced density matrix from $\ket{\Psi'(t)}$, and the $t$ parameter is omitted in the notation to avoid clutter. Such a replacement works because in the last lines of both Eq.~\eqref{eq:IABx} and Eq.~\eqref{eq:IBCy}, the instantaneous modular flows preceding $\ins_{X}$ ($X = D,DE,M,C$) do not change the reduced density matrix on $X$. This is because they have supports that either contain $X$ or are completely outside $X$. We remark that if $\rho'_{X}$ is non-invertible, to make sure $V(t)$ is a unitary, in Eq.~\eqref{eq:def-Vt} one should use the unitary extension $(\Tilde{\rho}'_{X})^{\pm \ii t}$ defined in App.~\ref{app:non-invertible}. 
  
    In the end, when we take state overlap between $\ins_{AB}(x)\ket{\Psi(t)}$ and $\ins_{BC}(y)\ket{\Psi(t)}$ to compute $\calF_J^{A,B,C}$, the unitary $V(t)$ gets cancelled, and we obtain Eq.~\eqref{eq:F-F'-app}. With Eq.~\eqref{eq:J-F}, we can conclude Eq.~\eqref{eq:J-J-psi'-app}. 

    {\it (Step 2)}: The second step is to show 
    \begin{equation}\label{eq:J-deform-app}
        J(A,B,C)_{\ket{\Psi'(t)}} = J(A,B,C')_{\ket{\Psi'(t)}}. 
    \end{equation}
    The derivation is crucially based on Markov properties of $\ket{\Psi'(t)}$ and the resulting consequence Eq.~\eqref{eq:K-markov-decomp-app} and Fact~\ref{fact:J-markov}. In the following computations, we shall use $K'_{\bullet}$ to denote the modular Hamiltonian obtained from $\ket{\Psi'(t)}$. 

    \begin{figure}[!h]
        \centering
        \includegraphics[width=0.5\linewidth]{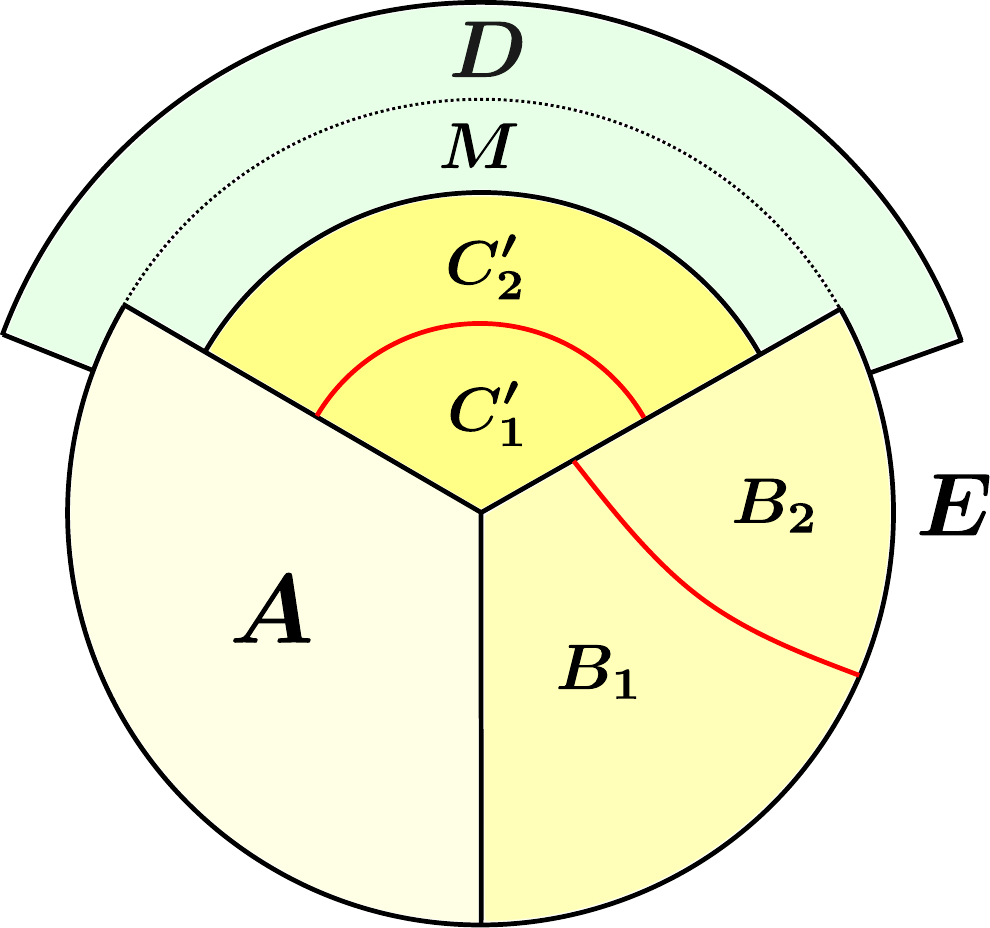}
        \caption{Region partitions $B = B_1B_2,C' = C'_1C'_2$, $C = C'M$.}
        \label{fig:reg-partition-app}
    \end{figure}
    We further make the partitions $B = B_1B_2$ and $C' = C'_1C'_2$ shown in Fig.~\ref{fig:reg-partition-app}, $\ket{\Psi'(t)}$ satisfies 
    \begin{equation}\label{eq:markovs-BC-app}
        I(B_1C'_1:M|B_2C'_2)_{\ket{\Psi'(t)}}=0. 
    \end{equation}
    This can be proved using strong subadditivity (SSA) \cite{Lieb:1973cp}. 
    \begin{equation}\label{eq:using-SSA-app}
        \begin{split}
            0 &\leq I(B_1C'_1:M|B_2C'_2)_{\ket{\Psi'(t)}} \\
            &\leq I(B_1C'_1:MD|B_2C'_2)_{\ket{\Psi'(t)}} \\ 
            & = I(B_1C'_1:MD|B_2C'_2)_{\ket{\Psi}} \\ 
            & = 0.
        \end{split}
    \end{equation}
    In Eq.~\eqref{eq:using-SSA-app}, the first and second inequality are due to SSA, and the equality in the third line is because the $\ket{\Psi'(t)} = \ins_{MD}(t)\ket{\Psi}$ and $\ket{\Psi}$ have the same entanglement entropies on $BC',B_2C'_2MD,B_2C'_2,BCD$, as the modular flow $\ins_{MD}$, which is just a unitary on $MD$, doesn't change the spectrum of the reduced density matrices on these regions. With Eq.~\eqref{eq:markovs-BC-app}, one can decompose 
    \begin{equation}
        K'_{BC}\ket{\Psi'(t)} = (K'_{BC'} + K'_{B_2C'_2M} - K'_{B_2C'_2})\ket{\Psi'(t)}.
     \end{equation}
     Plug this decomposition into the left-hand side of Eq.~\eqref{eq:J-deform-app}, we obtain 
     \begin{equation}\label{eq:J-deform-dev-app}
        \begin{split}
            &J(A,B,C)_{\ket{\Psi'(t)}} = J(A,B,C')_{\ket{\Psi'(t)}} \\
            +&\ii \bra{\Psi'(t)}[K'_{AB},K'_{B_2C'_2M}-K'_{B_2C'_2}]\ket{\Psi'(t)}  \\
            =& J(A,B,C')_{\ket{\Psi'(t)}}.
        \end{split}
     \end{equation}
    The vanishing of the second term in the right-hand side of Eq.~\eqref{eq:J-deform-dev-app} can be shown using the similar method: Notice $I(A:B_2|B_1)_{\ket{\Psi'(t)}} = I(A:B_2|B_1)_{\ket{\Psi}}=0$, we can decompose 
    \begin{equation}\label{eq:decomp-KAB}
        K'_{AB}\ket{\Psi'(t)} = (K'_{AB_1}+K'_{B}-K'_{B_1})\ket{\Psi'(t)}. 
    \end{equation}
    With this decomposition, the second line of Eq.~\eqref{eq:J-deform-dev-app} becomes 
    \begin{equation}
        \begin{split}
            &\ii \bra{\Psi'(t)}[K'_{AB},K'_{B_2C'_2M}-K'_{B_2C'_2}]\ket{\Psi'(t)} \\ 
            =&\ii \bra{\Psi'(t)}[K'_{B},K'_{B_2C'_2M}-K'_{B_2C'_2}]\ket{\Psi'(t)} \\ 
            =& J(B_1,B_2,C'_2M)_{\ket{\Psi'(t)}} - J(B_1,B_2,C'_2)_{\ket{\Psi'(t)}}=0, 
        \end{split}
    \end{equation} 
    where in the last line both modular commutators vanish because the Markov property 
    $I(B_1:C'_2M|B_2)_{\ket{\Psi'(t)}}=0,I(B_1:C'_2|B_2)_{\ket{\Psi'(t)}}=0$ [Fact~\ref{fact:J-markov}]. The first Markov condition can be proved using SSA in the same way as Eq.~\eqref{eq:using-SSA-app}: 
    \begin{equation}
        \begin{split}
            0\leq& I(B_1:C'_2M|B_2)_{\ket{\Psi'(t)}}  \\ 
            \leq& I(B_1:C'_2MD|B_2)_{\ket{\Psi'(t)}} \\ 
            =& I(B_1:C'_2MD|B_2)_{\ket{\Psi}} = 0,
        \end{split}
    \end{equation}
    and the second Markov condition holds simply because the reduced density matrix of $\ket{\Psi'(t)}$ on $AB$ is the same as the one from $\ket{\Psi}$. By now we finished derivation of Eq.~\eqref{eq:J-deform-app}.

    On $ABC'$, the reduced density matrix from $\ket{\Psi'(t)}$ is the same as the one from $\ket{\Psi}$. Therefore 
    \begin{equation}\begin{split}
    J(A,B,C')_{\ket{\Psi'(t)}} &= J(A,B,C')_{\ket{\Psi}}  \\ 
    &= J(A,B,C)_{\ket{\Psi}} \\ 
    &= \frac{\pi c_{-}}{3} ~,
\end{split}
\end{equation}
where in the second line, we used Fact~\ref{fact:J-def-inv}. 
\end{proof} 

{\it Remark:} We summarize the proof as follows: Step 1 is crucially based on the Markov decomposition move, which reflects {\bf P1} mentioned in the main text. In step 2, to deform the support of the modular commutator, we explicitly showed $\ket{\Psi'(t)}$ still satisfies several Markov conditions, which reflects a special case of {\bf P2}, i.e. those that directly follow from SSA. We shall provide a full proof of {\bf P2} in \cite{IMF-bulk}.

\section{Proof of Theorem~\ref{thm:no-go-electric-intro}} 
\label{app:proof-no-go-electric}
Before proving Theorem~\ref{thm:no-go-electric-intro}, we first prove the following lemma: 
\begin{lemma}\label{lemma:U-flow}
    Let $\ket{\psi}$ be a quantum state in a tensor product Hilbert space $\calH = \otimes_v \calH_v$. Any on-site unitary $U = \prod_v U_v$ can be moved freely in a sequence of instantaneous modular flows from $\ket{\psi}$: 
    \begin{equation}
    \begin{split}
        &U \ins_{A_n}(t_n) \cdots \ins_{A_1}(t_1) \ket{\psi}  \\ 
        =& \ins_{A_n}(t_n)\cdots \ins_{A_i}(t_i) U \ins_{A_{i-1}}(t_{i-1}) \cdots \ins_{A_1}(t_1) \ket{\psi} \\
        =& \ins_{A_n}(t_n) \cdots \ins_{A_1}(t_1) U\ket{\psi},\quad \forall i,~2<i<n. 
    \end{split}
    \end{equation}
\end{lemma}
\begin{proof}
    The proof only utilizes the on-site property of $U$ and the flipping move introduced in Sec.~\ref{sec:IMF}. 
    To prove lemma~\ref{lemma:U-flow}, we only need to show $U$ commutes with any $\ins_{A_i}(t_i)$ in the sequence. Let $\ket{\psi_i}$ denotes the current state for $\ins_{A_i}(t_i), \forall i,~1\leq i\leq n$. 
    \begin{equation}
        \begin{split}
            U\ins_{A_i}(t_i)\ket{\psi_i} &= U_{A_i} U_{\overline{A_i}}\ins_{A_i}(t_i)\ket{\psi_i} \\ 
            &=U_{A_i} \ins_{A_i}(t_i)U_{\overline{A_i}}\ket{\psi_i} \\ 
            &=U_{A_i} \ins_{\overline{A_i}}(t_i)U_{\overline{A_i}}\ket{\psi_i} \\
            &=\ins_{\overline{A_i}}(t_i)U_{A_i}U_{\overline{A_i}}\ket{\psi_i} \\
            & = \ins_{A_i}(t_i) U\ket{\psi_i}, 
        \end{split}
    \end{equation}
    where $U_{A_i} \equiv \prod_{v\in A_i} U_v$. 
\end{proof}
We remark on this lemma as follows: Firstly, this result illustrates the advantage of instantaneous modular flows interplaying with unitaries. Secondly, one can conclude that if $U$ is a symmetry of the state $\ket{\psi}$, then $U$ is also a symmetry for the state evolved by a sequence of instantaneous modular flows from $\ket\psi$.

We now prove Theorem~\ref{thm:no-go-electric-intro}. 
\begin{proof}
Consider a $\gU(1)$ symmetric state $\ket{\Psi}$ on a two-dimensional lattice in a tensor product Hilbert space with finite local dimension as described in the main text. (1) On one hand, because the local Hilbert space $\calH_X$ of any region $X$ with a finite number of sites is finite dimensional, the local charge operator $Q_X$ should have a bounded spectrum and $\langle Q_X^2\rangle$ should be finite. (2) On the other hand, with the conditions that $\ket{\Psi}$ satisfies bulk {\bf A1} and have non-zero $\sigma_{xy}$ from $\Sigma(A,B,C)$ on a tripartite region $ABC$ shown in Fig.~\ref{fig:bulk-A1-J} (Right), we can show that 
\begin{align}
    \langle Q_{BC}^2 \rangle (t) &\equiv \bra{\Psi(t)}Q_{BC}^2 \ket{\Psi(t)} \\ 
    &=\langle Q_{BC}^2 \rangle + 2\sigma_{xy} t,\quad t\in\mathbb{R} \label{eq:QBC-t}.
\end{align} 
where $\ket{\Psi(t)} = \ins_{AB}(t)\ket{\Psi}$. This implies $\langle Q_{BC}^2 \rangle (t)$ as a function of $t$ is unbounded in both directions. Since $Q_{BC}^2$ is a positive operator and $\langle Q_{BC}^2 \rangle (t)\geq 0$, we can conclude $\langle Q_{BC}^2 \rangle$ must be infinite and hence the spectrum of $Q_{BC}$ is unbounded, where $BC$ contains only finite number of sites. Therefore, we have obtained a contradiction.

The main goal of the derivation below is to show 
\begin{equation}\label{eq:Sig-Sig-t}
    \Sigma(A,B,C)_{\ket{\Psi(t)}} = \Sigma(A,B,C)_{\ket{\Psi}} = \sigma_{xy},
\end{equation}
with which one can easily solve the differential equation Eq.~\eqref{eq:dQ-t-app} 
\begin{equation}\label{eq:dQ-t-app}
    \frac{d}{dt} \langle Q_{BC}^2 \rangle (t) = 2\Sigma(A,B,C)_{\ket{\Psi(t)}},\quad \forall t\in\mathbb{R}.
\end{equation}
and obtain Eq.~\eqref{eq:QBC-t}.

The route of the derivation of Eq.~\eqref{eq:Sig-Sig-t} is quite similar to the derivation of Lemma~\ref{lemma:Jt} in App.~\ref{app:lemma-Jt}. 

{\it (Step 1):} Consider Fig.~\ref{fig:deform}. The first step is to show 
\begin{equation}\label{eq:Sigma-Sigma'}
    \Sigma(A,B,C)_{\ket{\Psi(t)}} = \Sigma(A,B,C)_{\ket{\Psi'(t)}},
\end{equation}
where $\ket{\Psi'(t)} = \ins_{MD}(t)\ket\Psi$. With Eq.~\eqref{eq:Sigma-F}, to show Eq.~\eqref{eq:Sigma-Sigma'}, we can actually show a stronger result 
\begin{equation}\label{eq:F-F'-sigma}
    \calF^{A,B,C}_{\Sigma}(x,y)_{\ket{\Psi(t)}} = \calF^{A,B,C}_{\Sigma}(x,y)_{\ket{\Psi'(t)}}.
\end{equation}
The left hand side of Eq.~\eqref{eq:F-F'-sigma} involves $\ins_{AB}(x)\ket{\Psi(t)}$ and $U_{BC}(y)\ket{\Psi(t)}$. We shall relate them to $\ins_{AB}(x)\ket{\Psi'(t)}$ and $U_{BC}(y)\ket{\Psi'(t)}$ by the same unitary $V(t)$ defined in Eq.~\eqref{eq:def-Vt}. 

For $\ins_{AB}(x)\ket{\Psi(t)}$, we can simply use the result Eq.~\eqref{eq:IABx-Vt}. For $U_{BC}(y)\ket{\Psi(t)}$, 
\begin{equation}\label{eq:UBCy-Vt}
    \begin{split}
        &U_{BC}(y)\ket{\Psi(t)} \\ 
        =& U_{BC}(y)\ins_{D}(-t) \ins_{DE}(t)\ins_{M}(-t)\ins_{C}(t)\ket{\Psi'(t)} \\
        =& \ins_{D}(-t) \ins_{DE}(t)\ins_{M}(-t)\ins_{C}(t)U_{BC}(y)\ket{\Psi'(t)}\\ 
        =& V(t)U_{BC}(y)\ket{\Psi'(t)}. 
    \end{split}
\end{equation}
The explanation of Eq.~\eqref{eq:UBCy-Vt} is the following: The first line is simply a substitution of $\ket{\Psi(t)}$ using Eq.~\eqref{eq:IAB-decompose}. From the second line to the third line, we utilized Lemma~\ref{lemma:U-flow}. In the last line, to see one can replace the instantaneous modular flows with $V(t)$ defined in Eq.~\eqref{eq:def-Vt}, notice $U_{BC}(y)$ doesn't change any of the reduced density matrices on $D,DE,M,C$, and then one can simply repeat the argument given in the paragraph above Eq.~\eqref{eq:def-Vt}. We remark again that if any reduced density matrix is non-invertible, one can simply use the unitary extension defined in App.~\ref{app:non-invertible} in $V(t)$.

Therefore, we can see Eq.~\eqref{eq:F-F'-sigma} holds as the unitary $V(t)$ is cancelled when we take state overlap between $\ins_{AB}(x)\ket{\Psi(t)}$ and $U_{BC}(y)\ket{\Psi(t)}$.

{\it (Step 2)}: The next step is to show 
\begin{equation}\label{eq:step2}
    \Sigma(A,B,C)_{\ket{\Psi'(t)}} = \Sigma(A,B,C')_{\ket{\Psi'(t)}}.
\end{equation}
With $Q_{BC} = Q_{B} + Q_{C'} + Q_{M}$, we obtain 
\begin{equation}\label{eq:Sigma-ABC'}
    \begin{split}
    \Sigma(A,B,C)_{\ket{\Psi'(t)}} &= \Sigma(A,B,C')_{\ket{\Psi'(t)}} \\ 
    &+ \ii \langle \Psi'(t)| [K'_{AB},Q_{B}Q_{M}] | \Psi'(t) \rangle \\ 
    & =  \Sigma(A,B,C')_{\ket{\Psi'(t)}}. 
\end{split}
\end{equation}
Now we shall show the second term in the right-hand of Eq.~\eqref{eq:Sigma-ABC'} vanishes. As a starter, we emphasize that as a consequence of Lemma~\ref{lemma:U-flow} $\ket{\Psi'(t)}$ is also symmetric under the same $\gU(1)$ operator generated by $Q = \sum_v Q_v$ as the one for $\ket{\Psi}$, therefore both Fact~\ref{fact:basic} and Lemma~\ref{lemma:Sigma-markov} are applicable for $\ket{\Psi'(t)}$. Firstly, we can use Fact~\ref{fact:basic} (2) to conclude $[K'_{AB},Q_{AB}Q_M]=0$. With $Q_{B} = Q_{AB}-Q_A$ and $[K'_{AB},Q_{AB}]=0$, we can write the second term in Eq.~\eqref{eq:Sigma-ABC'} as 
\begin{equation}\label{eq:Sigma-vanish}
    \begin{split}
        &\ii \langle \Psi'(t)| [K'_{AB},Q_{B}Q_{M}] | \Psi'(t) \rangle \\ 
        =&\ii  \langle \Psi'(t)| [K'_{AB},(Q_{AB}-Q_A)Q_{M}] | \Psi'(t) \rangle \\
        =& -\ii \langle \Psi'(t)| [K'_{AB},Q_{A}Q_{M}] | \Psi'(t) \rangle.  
    \end{split}
\end{equation}
Then, with the same partition $B = B_1B_2$ as Fig.~\ref{fig:reg-partition-app}, one can decompose $K'_{AB}\ket{\Psi'(t)}$ as in Eq.~\eqref{eq:decomp-KAB}. Plug it into Eq.~\eqref{eq:Sigma-vanish}, we can see that
\begin{equation}
    \begin{split}
        &\ii \langle \Psi'(t)| [K'_{AB},Q_{A}Q_{M}] | \Psi'(t) \rangle \\ 
        =& \ii \langle \Psi'(t)| [K'_{AB_1}+K'_B-K'_{B_1},Q_{A}Q_{M}] | \Psi'(t) \rangle \\ 
        =& \ii \langle \Psi'(t)| [K'_{AB_1},Q_{A}Q_{M}] | \Psi'(t) \rangle \\ 
        =&0,
    \end{split}
\end{equation}
where the last line is because of $I(M:B_1|A)_{\ket{\Psi'(t)}}=0$ [Lemma~\ref{lemma:Sigma-markov}]. This Markov condition can be shown using SSA in the same way as the derivation of Eq.~\eqref{eq:using-SSA-app}: 
\begin{equation}
    \begin{split}
        0&\leq I(M:B_1|A)_{\ket{\Psi'(t)}} \\ 
        &\leq I(MD:B_1|A)_{\ket{\Psi'(t)}} \\ 
        & = I(MD:B_1|A)_{\ket{\Psi}} = 0. 
    \end{split}
\end{equation}
By now, we have achieved the goal for the second step [Eq.~\eqref{eq:step2}].

Since the reduced density matrix for computing $\Sigma(A,B,C')_{\ket{\Psi'(t)}}$ is on $ABC'$ and the same as the one from $\ket{\Psi}$, we obtain 
\begin{equation}
    \begin{split}     \Sigma(A,B,C')_{\ket{\Psi'(t)}} &= \Sigma(A,B,C')_{\ket{\Psi}} \\ 
        &=\Sigma(A,B,C)_{\ket{\Psi}} \\ 
        & = \sigma_{xy},
    \end{split}
\end{equation}
where the deformation $C'\to C$ in the second line is due to the topological property [Lemma~\ref{lemma:topology-Sigma}]. Thus, we obtained Eq.~\eqref{eq:Sig-Sig-t} and finished the proof. 
\end{proof}

\section{Generalizations of the no-go theorems}
\label{app:general-thm} 

The no-go theorems~\ref{thm:no-go-thermal-intro} and~\ref{thm:no-go-electric-intro}  
both involve the finite local Hilbert space dimensions condition. This condition can be relaxed into  
the condition that $\ket{\Psi}$ has finite local entanglement entropies and finite local charge fluctuations (i.e. $(S_A)_{\ket{\Psi}}<\infty$ and $\langle \Psi | Q_A^2|\Psi \rangle - \langle \Psi | Q_A|\Psi \rangle^2 < \infty$ for any local region $A$).\footnote{We left it a question on if the finiteness of entropy implies the finiteness of charge fluctuation.}

{\it Preliminaries:} We shall consider a state $\ket{\Psi}$ with well-defined local reduced density matrices $\rho_A$ and modular Hamiltonians $K_A = -\ln\rho_A$, and their von Neumann entropies are finite $S(\rho_A)<\infty$, where $A$ stands for any regions with finite number of lattice sites. First of all, one can straightforwardly define instantaneous modular flows in the same manner as in the main text. We shall assume that the modular flow gives physical quantities as a smooth function of $t$. Moreover, the implications of the bulk {\bf A1} condition still holds, which we explicitly explain in the following:
\begin{itemize}
    \item Fact~\ref{fact:markov} still holds, since its proof \cite{shi2020fusion} only uses strong subadditivity of entanglement entropies, which is true regardless of the local Hilbert space dimension \cite{Lieb:1973cp}.
    \item Lemma~\ref{lemma:Markov-flow} is still valid and therefore, one can still do Markov decomposition move. This is because if $\ket{\Psi}$ has finite local entanglement entropies, even if it is in an infinite dimensional Hilbert space, the structure theorem Eq.~\eqref{eq:rho-ABC-dec} of a Markov state $\rho_{ABC}$ from $\ket{\Psi}$ still holds \cite{Jen2006}, and proof of Lemma~\ref{lemma:Markov-flow} given in App.~\ref{app:markov-flow} is still applicable.
    \item One can still decompose the modular Hamiltonian of a Markov state, namely if $I(A:C|B)_{\ket{\Psi}} =0 $, then 
    \begin{equation}\label{eq:KABC-dec-app}
        K_{ABC}\ket{\Psi} = (K_{AB}+K_{BC}-K_{B})\ket{\Psi},
    \end{equation}
    where $K_{\bullet}$ are modular Hamiltonians of $\rho_{\bullet}$ from $\ket{\Psi}$. This equation can be obtained by taking the first derivative with respect to $t$ at $t =0$ for the Markov flow equation $\ins_{ABC}(t)\ket{\Psi} = \ins_{AB}(t)\ins_{BC}(t)\ins_B(-t)\ket{\Psi},\forall t\in\mathbb{R}$. 
\end{itemize}
In summary, if $\ket{\Psi}$ satisfies bulk {\bf A1} and has finite local entanglement entropies, one can still do Markov decomposition move and decompose the modular Hamiltonian as in Eq.~\eqref{eq:KABC-dec-app} for a Markov state. 

Now we discuss the modified no-go theorems: 
\begin{theorem}\label{thm:no-go-thermal-general}
    There is no quantum many-body state $\ket{\Psi}$ in a tensor product Hilbert space with finite local entanglement entropies satisfying bulk {\bf A1} with non-zero $c_{-}$ from $J(A,B,C)_{\ket\Psi}$ using Eq.~\eqref{eq:J-c}. 
\end{theorem}
\begin{theorem}\label{thm:no-go-electric-general}
    There is no quantum many-body state $\ket{\Psi}$ with on-site $\gU(1)$ symmetry in a tensor product Hilbert space with finite local entanglement entropies and finite local charge fluctuations satisfying bulk {\bf A1} with non-zero $\sigma_{xy}$ from $\Sigma(A,B,C)_{\ket\Psi}$ using Eq.~\eqref{eq:Sigma-sigma}. 
\end{theorem}

The proofs of Theorem~\ref{thm:no-go-thermal-general} and Theorem~\ref{thm:no-go-electric-general} are vastly overlapped with the proofs of Theorem~\ref{thm:no-go-thermal-intro} and Theorem~\ref{thm:no-go-electric-intro}. Below we explain how the original proofs can be applied in this context: 

{\it Proof sketch of Theorem~\ref{thm:no-go-thermal-general}:} The proof is still by contradiction. Let $\ket{\Psi}$ be a state with all the properties in Theorem~\ref{thm:no-go-thermal-general}. 

(1) Firstly, Lemma~\ref{lemma:Jt} is still valid and the proof is the same as the one in App.~\ref{app:lemma-Jt}: The step 1 in the proof of Lemma~\ref{lemma:Jt} can be repeated because one can still apply Markov decomposition move, which is the only non-trivial ingredient. For step 2, the main ingredients are strong subadditivity and decomposition of modular Hamiltonian of a Markov state Eq.~\eqref{eq:KABC-dec-app}, which are also applicable as we discussed above, therefore the step 2 in the proof of Lemma~\ref{lemma:Jt} is also applicable. 

(2) Once Lemma~\ref{lemma:Jt} is proved, then one can repeat the same derivation mentioned in the main body of the proof of Theorem~\ref{thm:no-go-thermal-intro} in the main text, and conclude 
\begin{equation}
    (S_{BC})_{\ket{\Psi(t)}} = (S_{BC})_{\ket{\Psi}} + \frac{\pi c_{-}}{3}t,\quad \forall t\in\mathbb{R}.
\end{equation}
Since $(S_{BC})_{\ket{\Psi(t)}}\geq 0,\forall t\in\mathbb{R}$ and one can take $\pi c_{-}t/3\to -\infty$ by letting $t \to -\infty$ (or $+\infty$) if $c_{-}>0$ (or $<0$), $(S_{BC})_{\ket{\Psi}}$ must be infinite, which contradicts to the finite local entropies assumption of $\ket{\Psi}$. 

{\it Proof sketch of Theorem~\ref{thm:no-go-electric-general}:} The proof in App.~\ref{app:proof-no-go-electric} can be repeated for the same reasons given above, as it uses the same ingredients (Markov decomposition move and Eq.~\eqref{eq:KABC-dec-app} for Markov state). Similar to the proof above, once one repeats the derivations and obtains 
\begin{equation}
    \Sigma(A,B,C)_{\ket{\Psi(t)}} = \sigma_{xy},\quad \forall t \in\mathbb{R},
\end{equation}
where $\ket{\Psi(t)} = \ins_{AB}(t)\ket{\Psi}$. Moreover, one can show that the charge fluctuation on $\ket{\Psi(t)}$ on $BC$, defined as 
\begin{equation}
    \sigma_{BC}(t) \equiv \langle\Psi(t)|Q_{BC}^2|\Psi(t)\rangle - \langle\Psi(t)|Q_{BC}|\Psi(t)\rangle^2
\end{equation}
 satisfies 
 \begin{equation}
 \begin{aligned}
      &\frac{d \sigma_{BC}(t)}{dt} = 2\Sigma(A,B,C)_{\ket{\Psi(t)}} \\ 
      +& 2\ii \langle \Psi(t)|Q_{BC} |\Psi(t)\rangle \cdot \langle \Psi(t)|[K_{AB}(t),Q_{BC}] |\Psi(t)\rangle \\ 
       = &2\Sigma(A,B,C)_{\ket{\Psi(t)}},
 \end{aligned}
 \end{equation}
 where $K_{AB}(t)$ is the modular Hamitonian from $\ket{\Psi(t)}$ (which is equal to $K_{AB}$) and the term in the second line vanishes because 
 \begin{equation}
     \begin{split}
        & \langle \Psi(t)|[K_{AB}(t),Q_{BC}] |\Psi(t)\rangle  \\ 
        =& \langle \Psi(t)|[K_{AB}(t),Q_{B}] |\Psi(t)\rangle + \langle \Psi(t)|[K_{AB}(t),Q_{C}] |\Psi(t)\rangle \\ 
        =& \langle \Psi(t)|[K_{\overline{AB}}(t),Q_{B}] |\Psi(t)\rangle \\ 
        =&0.
     \end{split}
 \end{equation}
Therefore, we obtain 
\begin{equation}
    \sigma_{BC}(t) = \sigma_{BC} + 2\sigma_{xy}t,\quad \forall t\in\mathbb{R},
\end{equation}
where $\sigma_{BC}$ is the charge fluctuation computed on $\ket{\Psi}$. Since $\sigma_{BC}(t) \geq 0$ and one can take $2\sigma_{xy}t\to -\infty$ if $\sigma_{xy}\neq 0$, then $\sigma_{BC}$ must be infinite, which contradicts to the finite charge fluctuations condition.

\phantomsection\addcontentsline{toc}{section}{References}
\bibliographystyle{ucsd}
\bibliography{ref} 
\end{document}